\pdfoutput=1 

\documentclass[a4paper,ukenglish,cleveref, autoref, thm-restate,authorcolumns]{lipics-v2019}

\title{Combining Weak Distributive Laws: Application to Up-To Techniques} 

\titlerunning{Combining Weak Distributive Laws: Application to Up-To Techniques} 

\author{Alexandre Goy}{Université Paris-Saclay, CentraleSupélec, MICS, France}{alexandre.goy@centralesupelec.fr}{}{} 

\author{Daniela Petri{\c s}an}{Université Paris-Diderot, IRIF, France}{petrisan@irif.fr}{}{}

\authorrunning{A. Goy and D. Petri{\c s}an} 

\Copyright{A. Goy and D. Petrisan} 

\ccsdesc[500]{Theory of computation~Formal languages and automata theory} 

\keywords{Weak distributive law, Weak lifting, Powerset monad, Distribution monad, Trace semantics, Behavioural equivalence, Up-to techniques, Alternating automata} 

\category{} 

\relatedversion{} 

\supplement{}

\nolinenumbers 

\hideLIPIcs  

\EventEditors{}
\EventNoEds{2}
\EventLongTitle{}
\EventShortTitle{}
\EventAcronym{}
\EventYear{}
\EventDate{}
\EventLocation{}
\EventLogo{}
\SeriesVolume{}
\ArticleNo{}

\usepackage{tikz}
\usetikzlibrary{cd}

\usepackage[quotation,notion,diagnose line=false,protect quotation={tikzcd}]{knowledge}

\newrobustcmd\grf{\kl[\grf]{f_*}}
\newrobustcmd\gridX{\kl[\grf]{\mathit{id_X}_*}}
\newrobustcmd\EMSt{\kl[\EMSt]{\EM{\monSt}}}
\newrobustcmd\EMSwl{\kl[\EMSwl]{\EM{\monSwl}}}
\newrobustcmd\funUuSt{\kl[\funUuSt]{\funU^{\monSt}}}
\newrobustcmd\funFuSt{\kl[\funFuSt]{\funF^{\monSt}}}

\newrobustcmd\muTX{\kl[\muTX]{\mu^{\monT}_X}}
\newrobustcmd\muTTX{\kl[\muTTX]{\mu^{\monT}_{\funT X}}}
\newrobustcmd\muTSTX{\kl[\muTSTX]{\mu^{\monT}_{\funS\funT X}}}
\newrobustcmd\muTGSTX{\kl[\muTGSTX]{\mu^{\monT}_{\funG\funS\funT X}}}
\newrobustcmd\muTY{\kl[\muTY]{\mu^{\monT}_Y}}
\newrobustcmd\muPX{\kl[\muPX]{\mu^{\monP}_X}}
\newrobustcmd\muPS{\kl[\muPS]{\mu^{\monP}_S}}
\newrobustcmd\muDX{\kl[\muDX]{\mu^{\monD}_X}}
\newrobustcmd\muDY{\kl[\muDY]{\mu^{\monD}_Y}}
\newrobustcmd\muPfX{\kl[\muPfX]{\mu^{\monPf}_X}}
\newrobustcmd\muPPfx{\kl[\muPPfx]{\mu^{\monP}_{\funPf X}}}
\newrobustcmd\muPcX{\kl[\muPcX]{\mu^{\monPc}_X}}
\newrobustcmd\etaTX{\kl[\etaTX]{\eta^{\monT}_X}}
\newrobustcmd\etaTY{\kl[\etaTY]{\eta^{\monT}_Y}}
\newrobustcmd\etaTZ{\kl[\etaTZ]{\eta^{\monT}_Z}}
\newrobustcmd\etaDX{\kl[\etaDX]{\eta^{\monD}_X}}
\newrobustcmd\etaPX{\kl[\etaPX]{\eta^{\monP}_X}}
\newrobustcmd\etaPcX{\kl[\etaPcX]{\eta^{\monPc}_X}}
\newrobustcmd\etaSX{\kl[\etaSX]{\eta^{\monS}_X}}
\newrobustcmd\etaTPX{\kl[\etaTPX]{\eta^{\monT}_{\funP X}}}
\newrobustcmd\etaTSX{\kl[\etaTSX]{\eta^{\monT}_{\funS X}}}
\newrobustcmd\etaDPX{\kl[\etaDPX]{\eta^{\monT}_{\funD X}}}
\newrobustcmd\etaTTX{\kl[\etaTTX]{\eta^{\monT}_{\funT X}}}
\newrobustcmd\etaDz{\kl[\etaDz]{\eta^{\monD}_{\{0\}}}}
\newrobustcmd\etaDzu{\kl[\etaDzu]{\eta^{\monD}_{\{0,1\}}}}
\newrobustcmd\etaSXa{\kl[\etaSXa]{\eta^{\monSc}_{(X,a)}}}
\newrobustcmd\muSX{\kl[\muSX]{\mu^{\monS}_X}}
\newrobustcmd\WDL{\kl[\WDL]{\Rightarrow_w}}

\newcommand{\mb}[1]{\boldsymbol{\mathsf{#1}}}

\renewcommand{\sf}[1]{\mathsf{#1}}

\renewcommand{\phi}{\varphi}

\newcommand{\cat}[1]{\mathsf{#1}}
\newcommand{\fun}[1]{{#1}}
\newcommand{\mon}[1]{\mb{#1}}
\DeclareMathOperator{\id}{id}
\DeclareMathOperator{\supp}{\sf{supp}}
\DeclareMathOperator{\Id}{Id}

\newrobustcmd\Set{\kl[\Set]{\cat{Set}}}
\newrobustcmd\Rel{\kl[\Rel]{\cat{Rel}}}
\newrobustcmd\KHaus{\kl[\KHaus]{\cat{KHaus}}}
\newrobustcmd\catC{\kl[\catC]{\cat{C}}}

\newrobustcmd{\Kl}[1]{\kl[\Kl]{\cat{Kl}(#1)}}
\newrobustcmd{\EM}[1]{\kl[\EM]{\cat{EM}(#1)}}
\newrobustcmd{\Alg}[1]{\kl[\Alg]{\cat{Alg}(#1)}}
\newrobustcmd{\Coalg}[1]{\kl[\Coalg]{\cat{Coalg}(#1)}}
\newrobustcmd{\SEM}[1]{\kl[\SEM]{\cat{SEM}(#1)}}

\newrobustcmd{\unit}[1]{\kl[\eta]{\eta^{#1}}}
\newrobustcmd{\mult}[1]{\kl[\mu]{\mu^{#1}}}

\newrobustcmd\funB{\kl[\funB]{\fun{\beta}}}
\newrobustcmd\monB{\kl[\monB]{\mon{B}}}
\newrobustcmd{\etaB}{\kl[\etaB]{\unit{\monB}}}
\newrobustcmd{\muB}{\kl[\muB]{\mult{\monB}}}

\newrobustcmd\funD{\kl[\funD]{\fun{D}}}
\newrobustcmd\monD{\kl[\monD]{\mon{D}}}
\newrobustcmd{\etaD}{\kl[\etaD]{\unit{\monD}}}
\newrobustcmd{\muD}{\kl[\muD]{\mult{\monD}}}
\newrobustcmd\funUuD{\kl[\funUuD]{\funU^{\monD}}}
\newrobustcmd\funFuD{\kl[\funFuD]{\funF^{\monD}}}

\newrobustcmd\dld{\kl[\dld]{\delta}}
\newrobustcmd\wdld{\kl[\wdld]{\delta}}
\newrobustcmd\wdls{\kl[\wdls]{\sigma}}
\newrobustcmd\wdlt{\kl[\wdlt]{\tau}}
\newrobustcmd\wdll{\kl[\wdll]{\lambda}}

\newrobustcmd\funDt{\kl[\funDt]{\widetilde{\funD}}}

\newrobustcmd\funF{\kl[\funF]{\fun{F}}}
\newrobustcmd\funFt{\kl[\funFt]{\widetilde{\funF}}}

\newrobustcmd\funG{\kl[\funG]{\fun{G}}}
\newrobustcmd\monG{\kl[\monG]{\mon{G}}}
\newrobustcmd\etaG{\kl[\etaG]{\eta{G}}}
\newrobustcmd\muG{\kl[\muG]{\mu{G}}}
\newrobustcmd\funGl{\kl[\funGl]{\widehat{\funG}}}
\newrobustcmd\funGe{\kl[\funGe]{\overline{\funG}}}
\newrobustcmd\funGwl{\kl[\funGwl]{\widehat{\funG}}}
\newrobustcmd\funGwlb{\kl[\funGwlb]{\widecheck{\funG}}}

\newrobustcmd\funI{\kl[\funI]{\fun{I}}}
\newrobustcmd\IdC{\kl[\IdC]{\Id_{\catC}}}

\newrobustcmd\iotawl{\kl[\iotawl]{\iota}}
\newrobustcmd\iotawls{\kl[\iotawls]{\iota^{\wdls}}}
\newrobustcmd\iotawlt{\kl[\iotawlt]{\iota^{\wdlt}}}

\newrobustcmd\funK{\kl[\funK]{\fun{K}}}

\newrobustcmd\funM{\kl[\funM]{\fun{M}}}

\newrobustcmd\funP{\kl[\funP]{\fun{P}}}
\newrobustcmd\monP{\kl[\monP]{\mon{P}}}
\newrobustcmd{\etaP}{\kl[\etaP]{\unit{\monP}}}
\newrobustcmd{\muP}{\kl[\muP]{\mult{\monP}}}
\newrobustcmd\funFuP{\kl[\funFuPf]{\funF^{\monP}}}
\newrobustcmd\funUuP{\kl[\funFuPf]{\funU^{\monP}}}

\newrobustcmd\funPc{\kl[\funPc]{\fun{P}_{\!c}}}
\newrobustcmd\monPc{\kl[\monPc]{\mon{P}_{\!c}}}
\newrobustcmd\etaPc{\kl[\etaPc]{\unit{\monPc}}}
\newrobustcmd\muPc{\kl[\muPc]{\mult{\monPc}}}
\newrobustcmd\funFuPc{\kl[\FuPc]{\funF^{\monPc}}}
\newrobustcmd\funUuPc{\kl[\UuPc]{\funU^{\monPc}}}

\newrobustcmd\funPf{\kl[\funPf]{\fun{P}_{\!f}}}
\newrobustcmd\monPf{\kl[\monPf]{\mon{P}_{\!f}}}
\newrobustcmd\etaPf{\kl[\etaPf]{\unit{\monPf}}}
\newrobustcmd\muPf{\kl[\muPf]{\mult{\monPf}}}
\newrobustcmd\funPft{\kl[\funPft]{\widetilde{\funPf}}}
\newrobustcmd\funFuPf{\kl[\funFuPf]{\funF^{\monPf}}}
\newrobustcmd\funUuPf{\kl[\funUuPf]{\funU^{\monPf}}}

\newrobustcmd\funPt{\kl[\funPt]{\widetilde{\funP}}}
\newrobustcmd\monPt{\kl[\monPt]{\widetilde{\monP}}}
\newrobustcmd\funPwl{\kl[\funPwl]{\widehat{\funP}}}

\newrobustcmd\piwl{\kl[\piwl]{\pi}}
\newrobustcmd\piwls{\kl[\piwls]{\pi^{\wdls}}}
\newrobustcmd\piwlt{\kl[\piwlt]{\pi^{\wdlt}}}

\newrobustcmd\funR{\kl[\funR]{\fun{R}}}
\newrobustcmd\monR{\kl[\monR]{\mon{R}}}
\newrobustcmd\etaR{\kl[\etaR]{\unit{\monR}}}
\newrobustcmd\muR{\kl[\muR]{\mult{\monR}}}

\newrobustcmd\funS{\kl[\funS]{\fun{S}}}
\newrobustcmd\monS{\kl[\monS]{\mon{S}}}
\newrobustcmd{\etaS}{\kl[\etaS]{\unit{\monS}}}
\newrobustcmd{\muS}{\kl[\muS]{\mult{\monS}}}
\newrobustcmd\funFdS{\kl[\funFdS]{\funF_{\monS}}}
\newrobustcmd\funFuS{\kl[\funFuS]{\funF^{\monS}}}
\newrobustcmd\funUuS{\kl[\funUuS]{\funU^{\monS}}}

\newrobustcmd\funSl{\kl[\funSl]{\widehat{\funS}}}
\newrobustcmd\monSl{\kl[\monSl]{\widehat{\monS}}}
\newrobustcmd\etaSl{\kl[\etaSl]{\widehat{\etaS}}}
\newrobustcmd\muSl{\kl[\muSl]{\widehat{\muS}}}

\newrobustcmd\funSwl{\kl[\funSwl]{\widehat{\funS}}}
\newrobustcmd\monSwl{\kl[\monSwl]{\widehat{\monS}}}
\newrobustcmd\etaSwl{\kl[\etaSwl]{\widehat{\etaS}}}
\newrobustcmd\muSwl{\kl[\muSwl]{\widehat{\muS}}}

\newrobustcmd\funSc{\kl[\funSc]{\check{\funS}}}
\newrobustcmd\monSc{\kl[\monSc]{\check{\monS}}}
\newrobustcmd\etaSc{\kl[\etaSc]{\check{\etaS}}}
\newrobustcmd\muSc{\kl[\muSc]{\check{\muS}}}

\newrobustcmd\funSt{\kl[\funSt]{\widetilde{\funS}}}
\newrobustcmd\monSt{\kl[\monSt]{\widetilde{\monS}}}
\newrobustcmd{\etaSt}{\kl[\etaSt]{\widetilde{\etaS}}}
\newrobustcmd{\muSt}{\kl[\muSt]{\widetilde{\muS}}}

\newrobustcmd\funSTw{\kl[\funSTw]{\widetilde{\funS\funT}}}
\newrobustcmd\monSTw{\kl[\monSTw]{\widetilde{\monS\monT}}}
\newrobustcmd\muSTw{\kl[\muSTw]{\mult{\monSTw}}}

\newrobustcmd\funSTt{\kl[\funSTt]{\widetilde{\fun{ST}}}}
\newrobustcmd\monSTt{\kl[\monSTt]{\widetilde{\mon{ST}}}}

\newrobustcmd\funT{\kl[\funT]{\fun{T}}}
\newrobustcmd\monT{\kl[\monT]{\mon{T}}}
\newrobustcmd{\etaT}{\kl[\etaT]{\unit{\monT}}}
\newrobustcmd{\muT}{\kl[\muT]{\mult{\monT}}}
\newrobustcmd\funUuT{\kl[\funUuT]{\funU^{\monT}}}
\newrobustcmd\funUdT{\kl[\funUdT]{\funU_{\monT}}}
\newrobustcmd\funFdT{\kl[\funFdT]{\funF_{\monT}}}
\newrobustcmd\funFuT{\kl[\funFuT]{\funF^{\monT}}}

\newrobustcmd\funTc{\kl[\funTc]{\check{\funT}}}
\newrobustcmd\monTc{\kl[\monTc]{\check{\monT}}}
\newrobustcmd\etaTc{\kl[\etaTc]{\check{\etaT}}}
\newrobustcmd\muTc{\kl[\muTc]{\check{\muT}}}

\newrobustcmd\funTe{\kl[\funTe]{\overline{\funT}}}
\newrobustcmd\monTe{\kl[\monTe]{\overline{\monT}}}
\newrobustcmd\etaTe{\kl[\etaTe]{\overline{\etaT}}}
\newrobustcmd\muTe{\kl[\muTe]{\overline{\muT}}}

\newrobustcmd\funTwe{\kl[\funTwe]{\overline{\funT}}}
\newrobustcmd\monTwe{\kl[\monTwe]{\overline{\monT}}}
\newrobustcmd\etaTwe{\kl[\etaTwe]{\overline{\etaT}}}
\newrobustcmd\muTwe{\kl[\muTwe]{\overline{\muT}}}

\newrobustcmd\funTt{\kl[\funTt]{\widetilde{\funT}}}
\newrobustcmd\monTt{\kl[\monTt]{\widetilde{\monT}}}
\newrobustcmd{\etaTt}{\kl[\etaTt]{\widetilde{\etaT}}}
\newrobustcmd{\muTt}{\kl[\muTt]{\widetilde{\muT}}}

\newrobustcmd\funU{\kl[\funU]{\fun{U}}}

\newrobustcmd\funV{\kl[\funV]{\fun{V}}}
\newrobustcmd\monV{\kl[\monV]{\mon{V}}}
\newrobustcmd\etaV{\kl[\etaV]{\unit{\monV}}}
\newrobustcmd\muV{\kl[\muV]{\mult{\monV}}}

\newrobustcmd\funTwo[1]{\kl[\funTwo]{\fun{2}^{#1}}}

\newrobustcmd{\KlP}{\kl[\KlP]{\Kl{\monP}}}
\newrobustcmd{\KlT}{\kl[\KlT]{\Kl{\monT}}}
\newrobustcmd{\KlS}{\kl[\KlS]{\Kl{\monS}}}
\newrobustcmd{\EMP}{\kl[\EMP]{\EM{\monP}}}
\newrobustcmd{\EMD}{\kl[\EMD]{\EM{\monD}}}
\newrobustcmd{\EMT}{\kl[\EMT]{\EM{\monT}}}
\newrobustcmd{\EMB}{\kl[\EMB]{\EM{\monB}}}
\newrobustcmd{\EMS}{\kl[\EMS]{\EM{\monS}}}
\newrobustcmd{\SEMT}{\kl[\SEMT]{\SEM{\monT}}}
\newrobustcmd{\SEMD}{\kl[\SEMD]{\SEM{\monD}}}
\newrobustcmd{\EMPc}{\kl[\EMPc]{\EM{\monPc}}}
\newrobustcmd{\EMPf}{\kl[\EMPf]{\EM{\monPf}}}

\newrobustcmd{\chash}{\kl[\chash]{c^\#}}
\newrobustcmd{\monST}{\kl[\monST]\monS \sticktogether \monT}
\newrobustcmd{\monPP}{\kl[\monPP]\monP \sticktogether \monP}
\newrobustcmd{\monPD}{\kl[\monPD]\monP \sticktogether \monD}
\newrobustcmd{\monPB}{\kl[\monPB]\monP \sticktogether \monB}
\newrobustcmd{\muST}{\kl[\muST] \mu^{\monST}}
\newrobustcmd{\etaST}{\kl[\etaST] \eta^{\monST}}
\newrobustcmd{\monRS}{\kl[\monRS]\monR \sticktogether \monS}
\newrobustcmd{\monRST}{\kl[\monRST]\monR \sticktogether \monS \sticktogether \monT}

\newcommand{\sticktogether}{%
	\mkern-4.5mu
	\mathchoice{}{}{\mkern0.2mu}{\mkern0.5mu}%
}

\usepackage{scalerel,stackengine}
\stackMath
\newcommand\widecheck[1]{%
\savestack{\tmpbox}{\stretchto{%
  \scaleto{%
    \scalerel*[\widthof{\ensuremath{#1}}]{\kern-.6pt\bigwedge\kern-.6pt}%
    {\rule[-\textheight/2]{1ex}{\textheight}}
  }{\textheight}%
}{0.5ex}}%
\stackon[1pt]{#1}{\scalebox{-1}{\tmpbox}}%
}
\newrobustcmd{\unions}{\kl[\unions]{\sf{unions}}}

\newrobustcmd{\semAA}[1]{\kl[\semAA]{\left[ #1 \right]_{\sf{aa}}}}
\newrobustcmd{\semNDA}[1]{\kl[\semNDA]{\left[ #1 \right]_{\sf{nda}}}}
\newrobustcmd{\bisim}{\kl[\bisim]{\sim}}

\knowledge {notion}
	| \catC
	
\knowledge {notion}
	| identity morphism
	| \id
	
\knowledge {notion}
	| \Set

\knowledge {notion}
	| \KHaus
	
\knowledge {notion}
	| Kleisli category
	| \Kl
	| \KlS
	| \funFdS

\knowledge {notion}
	| functor
	| functors
	| \fun
	| \funG
	| \funF
	| \funU

\knowledge {notion}
	| identity functor
	| \Id
	| \IdC
	
\knowledge {notion}
	| covariant powerset
	| covariant powerset monad
	| \funP
	| powerset monad
	| \monP
	| \etaP
	| \etaPX
	| \muP
	| \muPX
	| \muPS
	| \muPPfx
	| \funFuP
	| \funUuP
	| \KlP
	| \EMP
	
\knowledge {notion}
	| Vietoris monad
	| \monV
	| \funV
	| \etaV
	| \muV

\knowledge {notion}
	| finite powerset
	| finite powerset functor
	| finite powerset monad
	| \funPf
	| \monPf
	| \etaPf
	| \muPf
	| \funFuPf
	| \funUuPf
	| \muPfX
	
\knowledge {notion}
	| finite distribution functor
	| \funD
	| finite distribution monad
	| \monD
	| \etaD
	| \etaDX
	| \etaDPX
	| \etaDz
	| \etaDzu
	| \muD
	| \muDX
	| \muDY
	
\knowledge {notion}
	| ultrafilter monad
	| ultrafilter functor
	| \funB
	| \monB
	| \etaB
	| \muB
	| \EMB
	
\knowledge {notion}
	| monad
	| monads
	| \mon
	| \funT
	| \monT
	| \funS
	| \monS
	| unit
	| \eta
	| \etaT
	| \etaTX
	| \etaTPX
	| \etaTSX
	| \etaTTX
	| \etaTY
	| \etaTZ
	| \etaS
	| \etaSXa
	| \etaSX
	| multiplication
	| \mu
	| \muT	
	| \muTX
	| \muTY
	| \muS
	| \muTSTX
	| \muTTX
	| \muTGSTX
	| \funR
	| \monR
	| \etaR
	| \muR
	| \muSX
	
\knowledge {notion}
	| contravariant powerset
	| \funTwo

\knowledge {notion}
	| coalgebra
	| coalgebras
	| coalgebra morphism
	| morphism of coalgebras
	| \Coalg
	
\knowledge {notion}
	| \Alg
	
\knowledge {notion}
	| Eilenberg-Moore algebra
	| Eilenberg-Moore algebras
	| algebra morphism
	| morphism of algebras
	| algebra
	| algebras
	| Eilenberg-Moore category
	| \EM
	| \EMT
	| \EMS
	| \funUuT
	| \funFuT
	| \funUuS
	| \funFuS
	
\knowledge {notion}
	| distributive law
	| distributive laws
	| \dld

\knowledge {notion}
	| weak distributive law
	| weak distributive laws
	| \wdld
	| \wdls
	| \wdlt
	| \wdll
	| \WDL
  
\knowledge {notion}
	| lifting
	| liftings
	| \funSl
	| \monSl
	| \etaSl
	| \muSl
	| \funGl
	
\knowledge {notion}
	| extension
	| extensions
	| \funTc
	| \monTc
	| \etaTc
	| \muTc
	| \funTe
	| \monTe
	| \etaTe
	| \muTe
	| \funGe
	
\knowledge {notion}
	| \funGT

\knowledge {notion}
	| weak lifting
	| weak liftings
	| \funSt
	| \monSt
	| \etaSt
	| \muSt
	| \monPt
	| \funPt
	| \monSwl
	| \funSwl
	| \etaSwl
	| \muSwl
	| \piwl
	| \iotawl
	| \EMSwl
	| \funGwl
	| \funGwlb
	| \piwls
	| \piwlt
	| \iotawls
	| \iotawlt

\knowledge {notion}
	| weak extension
	| weak extensions 
	| \funTt
	| \monTt
	| \etaTt
	| \muTt
	| \funPft
	| \EMPf
	| \monTwe
	| \funTwe
	| \etaTwe
	| \muTwe

\knowledge {notion}
	| category of sets and relations
	| \Rel
	
\knowledge {notion}
	| semialgebra
	| semialgebras
	| semialgebra morphism
	| semialgebra morphisms
	| \SEM
	| \SEMT
	| \SEMD
	
\knowledge {notion}
	| \funI
	| \funK

\knowledge {notion}
  | weakly cartesian@fun
  
\knowledge {notion}
	| weakly cartesian@nat

\knowledge {notion}
	| convex powerset monad	
	| \funPc
	| \monPc
	| \etaPc
	| \etaPcX
	| \muPc
	| \muPcX
	
\knowledge {notion}
	| \funDt
	
\knowledge {notion}
	| convex algebra
	| convex algebras
	| barycentric algebra
	| barycentric algebras
	
\knowledge {notion}
	| convex subset
	| convex subsets
	
\knowledge {notion}
	| \monST
	| \monRS
	| \monRST
	| \funST
	| \muST
	| \etaST

\knowledge {notion}
	| \monPP

\knowledge {notion}
	| \monPB
	
\knowledge {notion}
	| \monPD

\knowledge {notion}
	| \monSTt
	| \funSTt
	
\knowledge {notion}
	| alternating automaton
	| alternating automata
	
\knowledge {notion}
	| \funFt
	
\knowledge {notion}
	| \funPwl
	
\knowledge {notion}
	| \funM
	| machine functor

\knowledge {notion}
	| language
	| \semAA
	
\knowledge {notion}
	| \semNDA
	| \semNDA{-}
	
\knowledge {notion}
	| \unions
	
\knowledge {notion}
	| bisimulation
	| bisimulations
	
\knowledge {notion}
	| bisimilarity
	| Bisimilarity
	| bisimilar
	| \bisim
	
\knowledge {notion}
	| sound
	| soundness
	| Sound

\knowledge {notion}
	| up-to technique
	| up-to techniques

\knowledge {notion}
	| compatible
	| compatibility
	| Compatibility
	
\knowledge {notion}
	| bialgebra
	| bialgebras
	| bialgebraic

\knowledge {notion}
	| union closure
	| Union closure

\knowledge {notion}
	| convex hull
	| Convex hull
	
\knowledge {notion}
	| Yang-Baxter
\usepackage{adjustbox}

\begin{document}

\maketitle

\begin{abstract}
The coalgebraic modelling of alternating automata and of probabilistic automata has long been obstructed by the absence of distributive laws of the powerset monad over itself, respectively of the powerset monad over the finite distribution monad. This can be fixed using the framework of weak distributive laws. We extend this framework to the case when one of the monads is only a functor. We provide abstract compositionality results, a generalized determinization procedure, and systematic soundness of up-to techniques. Along the way, we apply these results to alternating automata as a motivating example. Another example is given by probabilistic automata, for which our results yield soundness of bisimulation up-to convex hull.
\end{abstract}

\section{Introduction}

Coalgebras have known great success in the abstract modelling of a
wide range of systems originating from computer science. The theory is
parametric and modular with respect to the base category and the type
functor. Compositionality concerns as well as generalization of major
landmarks such as the determinization constructions for
automata~\cite{SilvaBBR10} led the coalgebra community to make heavy
use of Beck's theory of distributive laws, which can be seen as a way
of composing monads, see for
e.g.~\cite{jacobs2017introduction,jacobs2006bialgebraic-review,klin2015coalgebraic}. Given
two monads $\monS$, $\monT$ modelling two branching types for a given
system, a distributive law is a natural transformation of type
$\funT\funS \Rightarrow \funS\funT$ that satisfies four coherence
diagrams. This simple swapping operation suffices to produce a monad
on the composite functor $\funS\funT$, for multiplication can now be
defined as
$\funS\funT\funS\funT \Rightarrow \funS\funS\funT\funT \Rightarrow
\funS\funT$. Notably, distributive laws have been extended to the case
when one monad structure and the two corresponding coherence diagrams
are suppressed. This slight alteration produces a powerful tool to
model the interplay between branching behaviour (represented by a
monad) and machine-like behaviour (represented by a plain functor).

The powerset monad $\monP$ and the finite distribution monad $\monD$
on the category of sets are amongst the most commonly used, as they
constitute the basic bricks for representing respectively
non-deterministic and probabilistic behaviour. In the recent years, the
community stumbled over the seemingly surprising fact that there is
neither a distributive law of type
$\monD\monP \Rightarrow \monP\monD$, see~\cite{Varacca03PhD}, nor one
of type $\monP\monP \Rightarrow \monP\monP$,
see~\cite{KLIN2018261}. Acknowledging that further such
impossibilities may arise, Zwart and Marsden 
recently provided general algebraic conditions that make distributive
laws unattainable~\cite{Zwart_Marsden}. One unpleasant impact of all these negative results
is that the coalgebraic study of alternating automata and
probabilistic automata must contend with workarounds to fit into the
theory properly~\cite{klin2015coalgebraic,Bon17,Bon19}. Alternating
automata are systems for which a transition consists of an existential
step and then a universal step --- making the composed $\funP\funP$ a
functor of choice to store transitions. Unfortunately, there is no
distributive law of powerset over itself and even no monad structure
on $\funP\funP$ at all~\cite{KLIN2018261}. Some substitute modellings
have been built to reason coalgebraically about alternating automata,
for example, by going back and forth to the category of
posets~\cite{bertrand2018coalgebraic}. Similarly, probabilistic
automata send a state to a set of distributions -- making them a
seemingly easy target for the composed $\funP\funD$. But again, the
lack of distributive law and even of possible monad structures on
$\funP \funD$ make their analysis much harder than expected.

It is not only generalized determinization that poses problems for
these systems. Another aspect where the theory is not very smooth
concerns the so called \emph{up-to techniques}.  Computing a
bisimulation between two systems can be tedious or even require an
infinite number of steps. Up-to techniques' stated objective is to
prune branches in the exploration of the state space. Originating from
a lattice-theoretic mindset, see e.g.~\cite{pous_sangiorgi_2011} for a
comprehensive account, they have proved hugely popular in the last few
years, partly on account of the impressive results of Bonchi and
Pous~\cite{bonchi2013checking} to accelerate bisimulation computation
on determinized automata. From a category-theoretic perspective,
generalized determinization, distributive laws, and the soundness of
up-to techniques are intrinsically linked, as explored
in~\cite{DBLP:conf/csl/BonchiPPR14}. In the absence of a distributive
law, we cannot reuse this compositionality results for alternating or
probabilistic automata.  A breakthrough occured in~\cite{Bon17}, where
Bonchi et al. do obtain a form of coalgebraic determinization of
probabilistic automata and use it to prove the soundness of the up-to
convex hull technique. However, some of the required constructions
have to be redone by hand, in a way that is disappointingly close to
the usual framework relying on distributive laws. Things \emph{almost}
work well, but not quite.


Coincidentally, a not-quite theory of distributive laws has been
brought
to light by Garner in a very recent paper~\cite{Gar19}. This theory
originated in the work of Street~\cite{Street_weak_laws} and Böhm~\cite{Bohm2010}, from which
Garner picked a particular set of axioms to the purpose of exhibiting
the Vietoris monad as
a canonical `almost' lifting of the powerset monad. The basic
observation is that in many cases, a not-quite distributive law
$\monT \monS \Rightarrow \monS \monT$ fails to be one because of one
specific coherence diagram, namely the one that states that that the
unit of $\monT$ is compatible with $\monS$. A weak distributive law is
defined as making the three other coherence diagrams commute. By a
well-rounded category-theoretic analysis, such laws are proved to
produce a distributive law-lifting-extension trinity similar to the
standard theory, as well as a monad structure that combines $\monS$
and $\monT$ but whose functor is \emph{not} $\funS\funT$. In a
previous paper~\cite{goypetrisan}, we prove that there is a weak
distributive law of type $\monD \monP \Rightarrow \monP
\monD$, 
and thus exhibit the convex powerset monad presented in~\cite{Bon17}
as a canonical weak lifting of the powerset monad to $\monD$-algebras.

In the present paper we continue our exploration of other applications
of weak distributive laws to the theory of systems modelled as
coalgebras. Our contributions are three-fold:

\textbf{Coalgebraic semantics for alternating automata.}
Adapting an example from Garner, we
point out that there is a weak distributive law
of type $\monP \monP \Rightarrow \monP \monP$. We use this to enlight that
the procedure for determinization of alternating
automata of~\cite{klin2015coalgebraic} is canonical in the sense
of weak distributive laws.
Notice however that alternating automata are coalgebras
for the functor $2\times(\funP\funP)^A$, and not just $\funP\funP$,
which leads us to our next contribution.

\textbf{Generalized determinization.} Secondly, we extend the theory
of generalized determinization via weak distributive laws to the
setting where one monad is replaced by the composition of a monad and a
functor --- only one coherence diagram is left in this case. In this
context, we provide a compositionality result inspired by the work of
Cheng on iterated distributive laws (Theorem~\ref{theo:compo}).

\textbf{Soundness of up-to techniques.} Once the category-theoretic
understanding of the determinization of alternating automata,
respectively of probabilistic automata is settled, we are ready to
tackle another application of weak distributive laws, namely to
proving the soundness of \emph{up-to techniques} and exploiting the
compostionality approach
of~\cite{DBLP:conf/csl/BonchiPPR14,DBLP:journals/acta/BonchiPPR17}.
We show that up-to techniques obtained via \emph{weak} distributive
laws are sound, on an equal basis with the ones of distributive
laws. If in~\cite{goypetrisan} the focus was set on the weak
distributive law $\monD\monP \Rightarrow \monP\monD$ and probabilistic
automata, we now concentrate on the weak distributive law
$\monP\monP \Rightarrow \monP\monP$ and alternating automata and the
compatibility of the associated up-to technique.
We also retrieve compatibility of up-to convex hull as stated
in~\cite{Bon17}.

\textbf{Synopsis.} Sections~\ref{sec:prerequisites}
and~\ref{sec:weakframework} consist of reminders about the standard
and weak theory of distributive laws, respectively. Section~\ref{sec:weakframework} ends with
a compositionality result for weak distributive laws. In
Section~\ref{sec:gendet}, we perform generalized determinization with
respect to multiple weak distributive laws, at once from a theoretical
viewpoint and on the case study of alternating
automata. Section~\ref{sec:upto} deals with compatibility of up-to
techniques: we prove that the standard (bialgebraic) method remains
valid in our case and derive the up-to techniques arising from
generalized determinization for alternating automata and probabilistic
automata.

\section{Prerequisites} \label{sec:prerequisites}

\subsection{Functors and Monads}

We hereby recall a few popular functors and monads on the category $""\Set""$ of sets and functions.

\begin{enumerate}
\item The ""finite powerset functor"" $\funP : \Set \to \Set$ maps a set $X$ to the set all of subsets of $X$ and maps a function $f : X \to Y$ to its direct image. It can be extended into a "monad" $\monP = (\funP,\etaP,\muP)$ where unit is $\etaPX(x) = \{x\}$ and multiplication is given by union. 
\item The ""finite distribution functor"" $\funD : \Set \to \Set$ maps a set $X$ to the set of finitely supported probability distributions on $X$. Given a function $f : X \to Y$, the function $\funD f$ maps a probability distribution $\phi \in \funD X$ to its pushforward measure with respect to $f$:
\begin{equation}
\funD f (\phi) = y \mapsto \phi[f^{-1}(\{y\})] = \sum_{x \in f^{-1}(\{y\})} \phi(x)
\end{equation}
The functor $\funD$ can be extended as well into a monad $\monD = (\funD,\etaD,\muD)$ where unit is taking the Dirac distribution $\etaDX(x) = \delta_x$ and multiplication is given by distribution flattening:
\begin{equation}
\muD(\Phi) = x \mapsto \sum_{\phi \in \funD X} \Phi(\phi) \phi(x)
\end{equation}
\item Assume $A$ is a fixed finite alphabet. The ""machine functor"" $\funM : \Set \to \Set$ maps a set $X$ to the set $2 \times X^A$ whose elements are pairs $(o,t)$ with $o \in \{0,1\}$ and $t : A \to X$. It maps a function $f : X \to Y$ to the function $\funM f : (o,t) \mapsto (o,a \mapsto f(t(a)))$
\end{enumerate}

\subsection{Distributive Laws}

In this section we recall Beck's framework of "distributive laws", "liftings" and "extensions"~\cite{beck:distributiveLaws}. Let $""\monT"" = (\funT, \etaT, \muT)$, $\monS = (\funS, \etaS, \muS)$ be monads on a category $""\catC""$.

\begin{definition}[Distributive Law]
A ""distributive law"" of type $\monT\monS \Rightarrow \monS\monT$ is a natural transformation $\dld : \funT\funS \Rightarrow \funS\funT$ making the following diagrams commute:
\begin{center}
\begin{tikzcd}
\funT\funT\funS \ar[r, "\funT\dld" above] \ar[d, "\muT\funS" left] \arrow[drr, phantom, "(\muT)"]
 & \funT\funS\funT \ar[r, "\dld\funT" above] & \funS\funT\funT \ar[d,"\funS\muT" right] & \funT\funS\funS \ar[r, "\dld\funS" above] \ar[d, "\funT\muS" left] \arrow[drr, phantom, "(\muS)"] & \funS\funT\funS \ar[r, "\funS\dld" above] & \funS\funS\funT \ar[d,"\muS\funT"] \\
\funT \funS \ar[rr, "\dld" below] & & \funS\funT & \funT\funS \ar[rr,"\dld" below] & & \funS\funT \\
\funT \funS \ar[rr, "\dld" above] & {} & \funS\funT & \funT\funS \ar[rr,"\dld" above] & {} & \funS\funT \\
& \funS \ar[ul, "\etaT\funS" below left] \ar[ur, "\funS\etaT" below right] \arrow[u, phantom, "(\etaT)", near end] & & & \funT \ar[ul, "\funT\etaS" below left] \ar[ur, "\etaS\funT" below right] \arrow[u, phantom, "(\etaS)", near end] & \\
\end{tikzcd}
\end{center}
\end{definition}

A ""lifting"" of $\monS$ on $\monT$ is a monad $\monSl : \EMT \to \EMT$ in the Eilenberg-Moore category of $\monT$ such that $\funUuT \monSl = \monS \funUuT$ (i.e. $\funUuT$ commutes with functor, unit and multiplication). An ""extension"" of $\monT$ on $\monS$ is a monad $\monTe : \KlS \to \KlS$ in the Kleisli category of $\monS$ such that $\monTe \funFdS = \funFdS \monT$. In contrast to the weaker notions seen in the next section, these ones will sometimes be called \emph{strong}.

\begin{proposition} \label{prop:bij}
There is a bijective correspondence between "distributive laws" of type $\monT\monS \Rightarrow \monS\monT$, "liftings" of $\monS$ on $\monT$ and "extensions" of $\monT$ on $\monS$.
\end{proposition}

Additionally, the existence of a distributive law of type $\monT\monS \Rightarrow \monS\monT$ yields a monad structure on the functor $\funS\funT$ defined by $\monST = (\funS\funT, \etaS\etaT, \muS\muT \circ \funS \dld \funT)$. Note that in order to distinguish this composite monad $\monST$ from the mere distributive-law-type notation $\monS\monT$, characters $\monS$ and $\monT$ are glued together. This framework can be restricted to the case where only one of the involved functors is a monad, with the required commutative diagrams adapting accordingly. A distributive law of type $\monT\funG \Rightarrow \funG\monT$ is a natural transformation of the obvious type such that the $(\muT)$ and $(\etaT)$ diagrams commute. Such distributive laws are also called $\sf{EM}$-laws because they correspond to liftings of $\funG$ on $\monT$ i.e. functors $\funGl : \EMT \to \EMT$ such that $\funUuT \funGl = \funG \funUuT$. A distributive law of type $\funG\monS \Rightarrow \monS\funG$ is a natural transformation of the obvious type such that the $(\muS)$ and $(\etaS)$ diagrams commute. Such distributive laws are also called $\sf{Kl}$-laws because they correspond to extensions of $\funG$ on $\monS$ i.e. functors $\funGe : \KlS \to \KlS$ such that $\funGe \funFdS = \funFdS \funG$.\\

In the case when $\monS$ is the ""powerset monad"" $\monP : \Set \to \Set$, it is a natural requirement to ask that an "extension" of $\monT$ on $\monP$ preserves the order structure obtained by identifying $\KlP$ as the category $""\Rel""$ of sets and relations. We have the following result from Barr \cite{barr1970relational}:

\begin{proposition} \label{prop:canonical}
There is a (necessarily unique) "extension" of $\monT$ on $\monP$ whose functor is a $2$-functor  $\Rel \to \Rel$ if and only if both following facts hold:
\begin{itemize}
\item the functor $\funT$ is ""weakly cartesian@@fun"", meaning that it preserves weak pullbacks, and
\item the natural transformations $\etaT$ and $\muT$ are ""weakly cartesian@@nat"", meaning that their naturality squares are weak pullbacks.
\end{itemize}
\end{proposition}
The unique "extension" of Proposition~\ref{prop:canonical}, along with the corresponding "lifting" and "distributive law", will be called \emph{canonical}.\\

\section{The Weak Framework} \label{sec:weakframework}

In various cases, there simply does not exist any distributive law between two monads. More importantly, the uncomfortable situation when there is no monad structure at all on the composite functor happens quite frequently: examples in the category $\Set$ include $\funP\funP$ \cite{klin2018iterated} and $\funP\funD$ \cite{Varacca03PhD}. Recently, Zwart and Marsden \cite{Zwart_Marsden} provided general theorems that forbid certain distributive laws to exist.

\subsection{Weak Distributive Laws}

Böhm \cite{Bohm2010} and Street \cite{street2009weak} observed that weakening the distributive law axioms in a clever way actually still make $\monS\monT$ almost a monad. 
A particular way of weakening axioms was then picked out by Garner \cite{Gar19}, who proved that the Vietoris functor on $\KHaus$ is almost a lifting of $\monP$ on the ""ultrafilter monad"" $\monB$. Let us be more precise by introducing properly the notions. The cheeky basic idea is to simply drop the axiom that is often causing trouble, namely the $(\etaT)$ diagram.

\begin{definition}[Weak Distributive Law]
A ""weak distributive law"" of type $\monT\monS \Rightarrow \monS\monT$ is a natural transformation $\wdld : \funT\funS \Rightarrow \funS\funT$ such that diagrams $(\muT)$, $(\muS)$ and $(\etaS)$ commute.
\end{definition}

\begin{definition}[Weak Lifting]
A ""weak lifting"" of $\monS$ on $\monT$ is a monad $\monSwl : \EMT \to \EMT$  along with two natural transformations $\piwl : \funS \funUuT \Rightarrow \funUuT \funSwl$, $\iotawl : \funUuT \funSwl \Rightarrow \funS \funUuT$ such that $\piwl \circ \iotawl = 1$ and the following diagrams commute:
\begin{center}
\begin{tikzcd}
\funUuT \funSwl \funSwl \ar[r, "\iotawl \funSwl" above] \ar[d, "\funUuT \muSwl" left] \ar[drr, phantom, "(\iotawl \mu)"] & \funS \funUuT \funSwl \ar[r, "\funS \iotawl" above] & \funS \funS \funUuT \ar[d, "\muS \funUuT" right] & \funS \funS \funUuT \ar[d, "\muS \funUuT" left] \ar[r, "\funS \piwl" above] \ar[drr, phantom, "(\piwl \mu)"] & \funS \funUuT \funSwl \ar[r, "\piwl \funSwl" above] & \funUuT \funSwl \funSwl \ar[d, "\funUuT \muSwl" right] \\
\funUuT \funSwl \ar[rr, "\iotawl" below] & & \funS \funUuT & \funS \funUuT \ar[rr, "\piwl" below] & & \funUuT \funSwl \\
\funUuT \funSwl \ar[rr, "\iotawl" above] & {} & \funS \funUuT & \funS \funUuT \ar[rr, "\piwl" above] & {} & \funUuT \funSwl \\
& \funUuT \ar[ul, "\funUuT \etaSwl" below left] \ar[ur, "\etaS \funUuT" below right] \arrow[u, phantom, "(\iotawl \eta)", near end] & & & \funUuT \ar[ul, "\etaS \funUuT" below left] \ar[ur, "\funUuT \etaSwl" below right] \arrow[u, phantom, "(\piwl \eta)", near end] & 
\end{tikzcd}
\end{center}
\end{definition}

\begin{definition}[Weak Extension]
A ""weak extension"" of $\monT$ on $\monS$ is a functor $\funTwe : \KlS \to \KlS$ along with a natural $\muTwe : \funTwe\funTwe \Rightarrow \funTwe$ such that $\funTwe \funFdS = \funFdS \funT$ and $\muTwe \funFdS = \funFdS \muT$.
\end{definition}

Proposition~\ref{prop:bij} has a counterpart in the weak framework:

\begin{proposition} \label{prop:weakbij}
There is a bijective correspondence between "weak distributive laws" of type $\monT\monS \Rightarrow \monS\monT$, "weak extensions" of $\monT$ on $\monS$ and, whenever idempotents split in the base category $\catC$, "weak liftings" of $\monS$ on $\monT$.
\end{proposition}
\begin{proof}
As this result is proved in \cite{Gar19}, we only consider the bit of the proof that will be useful in the remainder of the paper. Let $\monSwl$ be a "weak lifting" of $\monS$ on $\monT$. The corresponding "weak distributive law" is given by \begin{tikzcd}
\funT\funS \ar[r,"\funT\funS\etaT" above] & \funT\funS\funT \ar[r, "\funT \piwl\funFuT" above] & \funT \funUuT \funSwl \funFuT \ar[rr, "\funUuT \varepsilon^{\monT} \funSwl \funFuT" above] & & \funUuT \funSwl \funFuT \ar[r, "\iotawl\funFuT" above] & \funS\funT \end{tikzcd}, where $\varepsilon^{\monT} : \funFuT \funUuT \Rightarrow 1$ is the counit of the adjunction $\funFuT \dashv \funUuT$.
\end{proof}
Similarly to the strong framework, the existence of a "weak distributive law" of type $\monT\monS \Rightarrow \monS\monT$ allows to build a monad structure mixing $\monS$ and $\monT$ from the composite adjunction
\begin{center}
\begin{tikzcd}
\EMSwl \ar[r,yshift=-0.4em] & \EMT \ar[l,yshift=0.6em, "\bot" below] \ar[r, yshift = -0.4em] & \catC \ar[l, yshift = 0.6em, "\bot" below]
\end{tikzcd}
\end{center}

As in the strong framework, this monad will be denoted by $\monST$. Again, characters $\monS$ and $\monT$ are glued together to denote the composite monad. Note that, unlike the strong case, the functor of $\monST$ is \emph{not} $\funS\funT$. We will often abuse notation and identify $""\monST""$ with its underlying functor $\funUuT \funSwl \funFuT$. Note that there are natural transformations $\iotawl^* = \iotawl \funFuT : \monST \Rightarrow \funS\funT$ and $\piwl^* = \piwl \funFuT : \funS\funT \Rightarrow \monST$ such that $\piwl^* \circ \iotawl^* = 1$. We use as much as possible the notation $\monST$ to stress the fact that this monad is a kind of composition of $\monS$ and $\monT$.\\

Again, one can consider the case when one of the monads is a plain functor. A "weak distributive law" of type $\monT \funG \Rightarrow \funG \monT$ only has to make the sole $(\muT)$ diagram commute. If idempotents split in $\catC$, such laws are in bijection with "weak liftings" of $\funG$ on $\monT$ i.e. functors $\funGwl : \EMT \to \EMT$ coming with natural $\piwl : \funG \funUuT \Rightarrow \funUuT \funGwl$, $\iotawl : \funUuT \funGwl \Rightarrow \funG \funUuT$ such that $\piwl \circ \iotawl = 1$. Concerning the other type $\funG \monS \Rightarrow \monS \funG$, we refrain to define a weak notion of distributive law, as it would require $(\muS)$ and $(\etaS)$ to commute, hence bring nothing new in comparison with the strong framework.\\

As idempotents split in $\Set$, Barr's result fits very well into the weak framework:

\begin{proposition} \label{prop:wcanonical}
There is a (unique) "weak extension" of $\monT$ on $\monP$ whose functor is a $2$-functor  $\Rel \to \Rel$ if and only if both following facts hold:
\begin{itemize}
\item the functor $\funT$ is "weakly cartesian@@fun", and
\item the natural transformation $\muT$ is "weakly cartesian@@nat".
\end{itemize}
\end{proposition}
Whenever it exists, this "weak extension" and the corresponding "weak distributive law" and "weak lifting" will be described as \emph{canonical}.

\subsection{Examples}

We now give three important examples where the weak framework proves useful. In all three cases, there is no canonical distributive law because the unit $\etaT$ is not weakly cartesian. The gap is even bigger concerning (at least) the second and third example, because as said previously, there is no distributive law of type $\monP\monP \Rightarrow \monP\monP$ or $\monD\monP \Rightarrow \monP\monD$ at all.

\begin{example} \label{ex:pb}
The category of algebras for the ultrafilter monad $\EMB$ is isomorphic to the category of compact Hausdorff spaces $""\KHaus""$. There is a canonical weak distributive law of type $\monB\monP \Rightarrow \monP\monB$ whose "weak lifting" is the Vietoris monad on $\KHaus$ and such that $""\monPB""$ is the filter monad on $\Set$. This result is the original motivation of Garner's paper \cite{Gar19}.
\end{example}

Garner also shows in Lemma~$17$ of \cite{Gar19} that there is a canonical "weak distributive law" of type $\monPf \monP \Rightarrow \monP \monPf$, where $\monPf$ is the "finite powerset monad". For symmetry purposes we make use of the variant he also mentions: 

\begin{example} \label{ex:pp}
There is a canonical "weak distributive law" of type $\monP \monP \Rightarrow \monP \monP$ defined by
\begin{equation} \wdld_X(\mathcal{A}) = \left\{B \subseteq X \mid B \subseteq \bigcup \mathcal{A} \text{ and } \forall A \in \mathcal{A}, A \cap B \neq \emptyset \right\} 
\end{equation}
Let us give an expression of the corresponding "weak lifting" of $\monP$ on $\monP$. The category $\EMP$ is isomorphic to the category of complete join semi-lattices. Let $(X,\sqcup)$ be an object. The underlying set of $""\funPwl"" (X,\sqcup)$ is $\Omega = \{A \in \funP X \mid A \text{ is stable under non-empty } \sqcup\}$, and its join is given for every $\mathcal{A} \in \funP \Omega$ by
\begin{equation}
\bigsqcup \mathcal{U} = \{\sqcup \{x_U \mid U \in \mathcal{U}\} \mid \forall U \in \mathcal{U}, x_U \in U\}
\end{equation}
We also have a monad $""\monPP""$ on $\Set$ whose functor maps a set $X$ to the set of all subsets of $X$ closed under non-empty union. On functions, it takes direct images twice.
\end{example}

\begin{remark}
From the point of view of logic, the transformation performed by the $\wdld$ of Example~\ref{ex:pp} amounts to transforming a conjunctive normal form into an equivalent disjunctive normal form. Indeed, consider that $X$ is a set of propositional variables. Seeing $\mathcal{U} \in \funP\funP X$ as a CNF and $\wdld_X(\mathcal{U})$ as a DNF, straightforward computations show that
\begin{equation}
\bigwedge_{U \in \mathcal{U}} \bigvee_{x \in U} x \equiv \bigvee_{V \in \wdld_X(\mathcal{U})} \bigwedge_{x \in V} x
\end{equation}
\end{remark}

\begin{example} \label{ex:pd}
There is a canonical "weak distributive law" of type $\monD \monP \Rightarrow \monP \monD$ defined by
\begin{align}
\wdld_X\left(\sum_i p_i A_i\right) = \left\{ \sum_i p_i \phi_i \mid \forall i, \supp(\phi_i) \subseteq A_i \right\}
\end{align}
where we use the formal sum notation with distinct $A_i$ and positive $p_i$. The corresponding weak lifting \cite{goypetrisan} is the ""convex powerset monad"" $\monPc$ described in detail in \cite{Bon17} and in a slightly different way in \cite{Bon19}. The monad $""\monPD""$ on $\Set$ is the convex sets of distributions monad, denoted by $C$ in \cite{Bon17}.
\end{example}

\subsection{Compositionality}

Results about distributive law composition \cite{cheng2011iterated} can be adapted to the weak framework to a limited extent.

\begin{theorem} \label{theo:compo}
Let $\monT$, $\monS$ be monads on a category $\catC$ in which idempotents split. Let $\funG$ be an endofunctor on $\catC$. Let $\wdld : \monT\monS \Rightarrow \monS\monT$, $\wdls : \monT\funG \Rightarrow \funG\monT$, $\wdlt : \monS\funG \Rightarrow \funG\monS$ be "weak distributive laws". Assume the so-called ""Yang-Baxter"" diagram commutes:
\begin{center}
\begin{tikzcd}
                                                                   & \funS\funT\funG \arrow[r, "\funS \wdls"]  & \funS\funG\funT \arrow[rd, "\wdlt \funT"] &                 \\
\funT\funS\funG \arrow[ru, "\wdld\funG"] \arrow[rd, "\funT\wdlt"'] &                                           &                                           & \funG\funS\funT \\
                                                                   & \funT\funG\funS \arrow[r, "\wdls \funS"'] & \funG\funT\funS \arrow[ru, "\funG\wdld"'] &                
\end{tikzcd}
\end{center}
Then the composite $\wdll = \begin{tikzcd} \monST \funG \ar[r,"\iotawl^* \funG" above] & \funS\funT\funG \ar[r, "\funS\wdls" above] & \funS\funG\funT \ar[r,"\wdlt\funT" above] & \funG\funS\funT \ar[r, "\funG\piwl^*" above] & \funG\monST \end{tikzcd}$ is a "weak distributive law" of type $\monST \funG \Rightarrow \funG \monST$. Moreover, if $\wdls$ and $\wdlt$ are strong, then $\wdll$ is strong.
\end{theorem}
A similar result holds if types are $\wdls : \funG\monT \Rightarrow \monT\funG$ and $\wdlt : \funG\monS \Rightarrow \monS\funG$. However, we believe that the results of \cite{cheng2011iterated} can not be adapted to the case where $\funG$ is replaced with a monad $\monR$, because the $(\muR)$ diagram for $\monST$ is unlikely to commute.

\begin{example} \label{ex:compo}
Let $\catC = \Set$, $\monT = \monS = \monP$ and let $\funG = \funM$ be the "machine functor" $2 \times (-)^A$ with respect to a finite alphabet $A$. Consider the "weak distributive law" $\wdld : \monP \monP \Rightarrow \monP \monP$ given in Example~\ref{ex:pp} and the renowned (strong) "distributive laws" $\sigma,\tau : \monP \funM \Rightarrow \funM \monP$ defined as
\begin{align} 
& \sigma_X(S) = \left( \bigwedge_{(o,f) \in S} o , a \mapsto \bigcup_{(o,f) \in S} f(a) \right)
& \tau_X(S) = \left( \bigvee_{(o,f) \in S} o, a \mapsto \bigcup_{(o,f) \in S} f(a) \right)
\end{align}
These laws satisfy the "Yang-Baxter" condition. Hence, there is a (strong) "distributive law" $\lambda : \monPP(2 \times -^A) \Rightarrow 2 \times \monPP^A$ given for any $U \in \monPP X$ by
\begin{equation}
\lambda_X(U) = \left( \bigvee_{S \in U} \bigwedge_{(o,f) \in S} o, a \mapsto \unions\left(\{\{f(a) \mid (o,f) \in S\} \mid S \in U\}\right)\right)
\end{equation}
where $""\unions"" = \piwl^* : \funP\funP \Rightarrow \monPP$ denotes closure of a set of sets under non-empty unions.
\end{example}

\section{Generalized Determinization} \label{sec:gendet}

In~\cite{klin2015coalgebraic}, Klin and Rot perform a powerset construction on "alternating automata", turning them into non-deterministic automata. To this purpose, a first version of the paper used to introduce a transformation that was wrongly identified as a "distributive law". Spotting the mistake sparked the chase of a correct distributive law of type $\monP \monP \Rightarrow \monP \monP$ -- chase which ended brutally with the result of Klin and Salamanca~\cite{klin2018iterated} that there can be no such law. In the corrected version of~\cite{klin2015coalgebraic}, the authors introduce an other natural transformation of the same type which happens to be the $\wdld$ of Example~\ref{ex:pp}. Thanks to this $\wdld$ they manage to correctly turn "alternating automata" into equivalent non-deterministic automata. However, this construction is not standard and it is still unclear how it would relate to the generalized determinization of coalgebras that heavily relies on "distributive laws". In this section, we show that the powerset construction of~\cite{klin2015coalgebraic} is an instance of a mild extension of the generalized determinization procedure described in~\cite{goypetrisan} with respect to "weak distributive laws".

\subsection{Determinization Procedure} \label{subsec:proc}

In our recent paper~\cite{goypetrisan}, the generalized determinization process for coalgebras is adapted to case of monad-monad "weak distributive laws". It is immediate that this generalized determinization stills works fine with monad-functor laws. Moreover, compositionality plays it role and modelling systems involving more than two monads and functors does not raise any difficulties. To keep things simple, and having in mind that we aim at modelling alternating automata, we will focus on a special case involving two monads and one functor. The following result can be easily inferred from the constructions of Lemma~$5.1$ of~\cite{goypetrisan}.

\begin{proposition} \label{prop:gendet}
Let $\monT$, $\monS$ be monads on a category $\catC$ in which idempotents split. Let $\funG$ be an endofunctor on $\catC$. Let $\wdld : \monT\monS \Rightarrow \monS\monT$, $\wdls : \monT\funG \Rightarrow \funG\monT$ be "weak distributive laws" and $\monSwl,\funGwl$ the corresponding "weak liftings". Then we have the following determinization diagram.
\begin{center}
\begin{tikzcd} \label{diag:gendet}
\Coalg{\funG\funS\funT} \ar[r,"\widehat{\funFuT}" above] \ar[d] & \Coalg{\funGwl\funSwl} \ar[r,"\widehat{\funUuT}" above] \ar[d] & \Coalg{\funG\funS} \ar[d] \\
\catC \ar[r,"\funFuT" below] & \EMT \ar[r,"\funUuT" below] & \catC
\end{tikzcd}
\end{center}
\end{proposition}

Moreover, it is interesting (and was not already remarked in~\cite{goypetrisan}) to note that in the above diagram, the functor $\Coalg{\funG\funS\funT} \to \Coalg{\funG\funS}$ can be expressed with the very same formula as in the case of "distributive laws".

\begin{lemma} \label{lem:same}
Let $(X,c)$ be a $\funG\funS\funT$-coalgebra. Then
\begin{equation}
\begin{tikzcd}[row sep=small] \label{eq:samething}
\widehat{\funUuT} \widehat{\funFuT}(X,c) = \funT X \ar[r,"\funT c" above] & \funT \funG \funS \funT X \ar[r, "\wdls_{\funS\funT X}" above] & \funG \funT \funS \funT X \ar[r, "\funG \wdld_{\funT X}" above] & \funG \funS \funT \funT X \ar[r, "\funG\funS\muTX" above] & \funG \funS \funT X
\end{tikzcd}
\end{equation}
\end{lemma}

\subsection{Application to Alternating Automata} \label{subsec:detaa}

The spirit of "alternating automata" originates in \cite{chandra1981lj}. It is well-suited to make systems deal with $\forall/\exists$ alternation --- a leitmotiv in logic. In this section, we give a coalgebraic modelling to alternating automata and provide their generalized determinization. This construction could have been equally performed on probabilistic automata to retrieve their determinized belief-state transformer. For this reason, notations and examples are kept close to the ones of~\cite{Bon17}.

\begin{definition}[Alternating automaton]
An ""alternating automaton"" is a tuple $(X,A,F,\to)$ where $X$ is a set of states, $A$ is a set of action labels, $F \subseteq X$ is a set of final states and $\to \subseteq X \times A \times \funP(X)$ is the transition relation. We will denote $(x,a,U) \in \to$ by $x \overset{a}{\to} U$.
\end{definition}

As often happens, the structure of "alternating automata" can be rephrased coalgebraically.

\begin{proposition}
An "alternating automaton" $(X,A,F,\to)$ can be identified with a coalgebra $c = \langle o, t \rangle : X \to 2 \times (\funP\funP X)^A$ on $\Set$, where $o(x) = 1$ iff $x \in F$ and $U \in t(x)(a)$ iff $x \overset{a}{\to} U$.
\end{proposition}

The "alternating automaton" of Figure~\ref{fig:aa} is directly inspired from the example given in \cite{Bon17} in order to highlight the vivid similarities between determinization of a probabilistic automaton into a belief-state transformer and determinization of alternating automaton into a non-deterministic automaton.

The "language" map $\semAA{-} : X \to 2^{A^*}$ of an "alternating automaton" $c = \langle o,t \rangle$ is
\begin{align}
& \semAA{x}(\varepsilon) = o(x)
& \semAA{x}(aw) = \bigvee_{U \in t_a(x)} \bigwedge_{y \in U} \semAA{y}(w)
\end{align}

\begin{figure}[h] \centering
\begin{tikzcd}
                                              &                                 & x_3                                                &                                            &  &                                                             &                                          & y_3                                                \\
                                              & {} \arrow[ru, dotted]           &                                                    &                                            &  &                                                             & {} \arrow[ru, dotted]                    & {} \arrow[u, dotted]                               \\
x_0 \arrow[r, "a"] \arrow[ru, "a", bend left] & {} \arrow[r, dotted]            & x_1 \arrow[r, "a"] \arrow[d, "a"', bend right]     & {} \arrow[luu, dotted] \arrow[ldd, dotted] &  & y_0 \arrow[ru, "a"', bend left] \arrow[rd, "a", bend right] &                                          & y_1 \arrow[u, "a"] \arrow[d, "a"', bend right]     \\
                                              &                                 & {} \arrow[u, dotted, bend right] \arrow[d, dotted] &                                            &  &                                                             & {} \arrow[rd, dotted] \arrow[ru, dotted] & {} \arrow[d, dotted] \arrow[u, dotted, bend right] \\
                                              & {} \arrow[r, dotted, bend left] & x_2 \arrow[l, "a", bend left]                      &                                            &  &                                                             & {} \arrow[r, dotted, bend left]          & y_2 \arrow[l, "a", bend left]                     
\end{tikzcd}
\caption{An alternating automaton $c_0$ on the alphabet $A = \{a\}$ with no final states. Solid lines denote existential transitions and dotted lines denote universal transitions. In other terms, any $U \in t_a(x)$ gives rise to one solid line starting from $x$, from the end of which starts one dotted line per element $t \in U$. For instance, $t_a(x_1) = \{\{x_1,x_2\},\{x_2,x_3\}\}$.}
\label{fig:aa}
\end{figure}
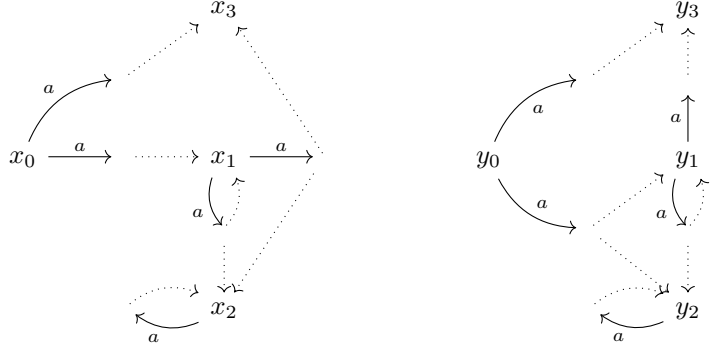

Because of the lack of "distributive law" of type $\monP \monP \Rightarrow \monP \monP$, in the past few years this modelling has been neglected in favor of workarounds e.g. using the category of posets \cite{bertrand2018coalgebraic}. Using the canonical "weak distributive law" $\wdld : \monP \monP \Rightarrow \monP \monP$ and the determinization procedure of Section~\ref{subsec:proc}, we can take a fresh look at determinization of "alternating automata". Note that performing once a powerset construction on $c$ yields a non-deterministic automaton $c^+$. We will still call $c^+$ the \emph{determinized} of $c$.\\

Let us apply Proposition~\ref{prop:gendet} with $\catC = \Set$, $\monS = \monT = \monP$, $\funG = \funM$. Consider $\wdld$ the canonical "weak distributive law" of type $\monP \monP \Rightarrow \monP \monP$ defined in Example~\ref{ex:pp} and $\wdls, \wdlt$ the two "distributive laws" of type $\monP \funM \Rightarrow \funM \monP$ defined in Example~\ref{ex:compo}. Let $c : X \to \funM\funP\funP X$ be an "alternating automaton". Remind that $c = \langle o, t\rangle$ with $o : X \to 2$ and $t : X \to (\funP\funP X)^A$. We use the convenient notation $t_a(x) = t(x)(a)$ and hereby determinize $c$ with respect to the inner powerset. Thanks to Lemma~\ref{lem:same}, this amounts to computing
\begin{equation}
\begin{tikzcd}
c^+ = \funP X \ar[r, "\funP c" above] & \funP \funM \funP \funP X \ar[r, "\wdls_{\funP\funP X}" above] & \funM \funP \funP \funP X \ar[r, "\funM \wdld_{\funP X}" above] & \funM \funP \funP \funP X \ar[r, "\funM\funP \muPX" above] & \funM \funP \funP X
\end{tikzcd}
\end{equation}
Then $c^+ = \langle o^+,t^+\rangle$ maps $U$ to $\left(\bigwedge_{x \in U} o(x), a \mapsto \left\{\bigcup_{x \in U} K_x \mid \forall x \in U, K_x \in \unions(t_a(x))\right\} \right)$.
As noted in~\cite{klin2018iterated}, this non-standard determinization of "alternating automata" may not be efficient because many states are added on account of $\unions$. However, our generic framework also comes with systematic compatibility of up-to techniques that will fix this issue. Let $""\semNDA{-}"" : \funP X \to 2^{A^*}$ be the usual non-deterministic semantics of the determinized $c^+ : \funP X \to \funM \funP \funP X$. As remarked previously in~\cite{klin2018iterated}:

\begin{proposition} \label{prop:goodsemantics}
For all $(w,U) \in A^* \times \funP X$, $\semNDA{U}(w) = \bigwedge_{x \in U} \semAA{x}(w)$.\\
In particular, equation $\semNDA{-} \circ \etaPX = \semAA{-}$ holds.
\end{proposition}

\begin{remark} \label{rem:chain}
Note that by applying multiple times Proposition~\ref{prop:gendet}-like results, one can actually chain determinizations. For instance, provided a "weak distributive law" of type $\monS \funG \Rightarrow \funG \monS$, one can further determinize $c^+ : \funT X \Rightarrow \funG \funS \funT X$ by applying Proposition~\ref{prop:gendet} by formally replacing $(\monS,\monT,\wdld,\wdls)$ with $(\mon{I},\monS,1,\wdlt)$ where $\mon{I}$ is the identity monad. In the case of an "alternating automaton" with $n$ states and the $\wdlt$ of Example~\ref{ex:compo}, this amounts to performing an additional \emph{standard} powerset construction, resulting in a standard deterministic automaton with $2^{2^n}$ states.
\end{remark}

\begin{figure} \centering
\begin{tikzcd}
              & \{x_0\} \arrow[ld, "a"', bend right] \arrow[d, "a"] \arrow[rd, "a", bend left]                 &                   &         & \{y_0\} \arrow[d, "a"] \arrow[ld, "a"', bend right] \arrow[rd, "a", bend left]                 &                   \\
\{x_3\}       & \{x_1\} \arrow[d, "a"] \arrow[ld, "a"', bend right] \arrow[rd, "a", bend left]                 & {\{x_1,x_3\}}     & \{y_3\} & {\{y_1,y_2\}} \arrow[d, "a"] \arrow["a"', loop, distance=2em, in=215, out=145] \arrow[r, "a"'] & {\{y_1,y_2,y_3\}} \\
{\{x_2,x_3\}} & {\{x_1,x_2\}} \arrow[l, "a"] \arrow[r, "a"'] \arrow["a"', loop, distance=2em, in=305, out=235] & {\{x_1,x_2,x_3\}} &         & {\{y_2,y_3\}}                                                                                  &                  
\end{tikzcd}
\caption{A portion of the non-deterministic automaton $c_0^+$ obtained by determinizing once the alternating automaton $c_0$ of Figure~\ref{fig:aa}.}
\label{fig:aadet}
\end{figure}
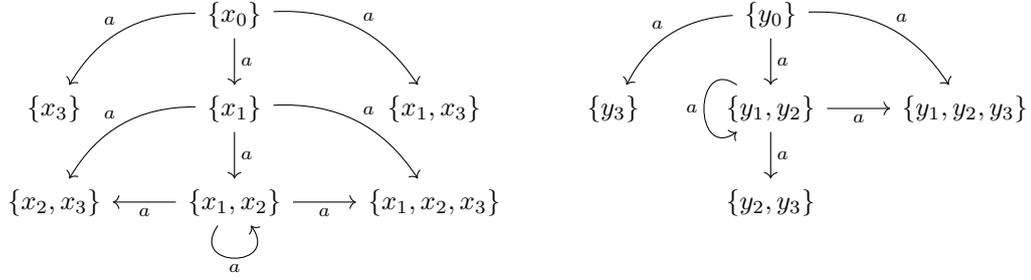

\section{Up-To Techniques} \label{sec:upto}

This section is dedicated to derive and use abstract results about up-to techniques in the case of "weak distributive laws". Such techniques have already been obtained in some cases where there is no "distributive law" by redoing manually all the standard proofs, see e.g. \cite{Bon17} for the "convex powerset monad". By identifying the presence of a "weak distributive law", we are able to easily adapt the known results concerning up-to techniques.

\subsection{Context and Congruence Closure}

Let $\funT,\funF$ be endofunctors on $\Set$. We consider a system given by an $\funF$-coalgebra $\xi : X \to \funF X$. For any relation $R \subseteq  X \times X$, let
\begin{equation}
\sf{Rel}(\funF)_{\xi}(R) = \{(u,v) \in X \times X \mid \exists t \in \funF R, (\xi(u),\xi(v)) = (\funF\pi_1(t),\funF\pi_2(t))\}
\end{equation}
where $\pi_1,\pi_2 : R \to X$ are the canonical projections. A relation $R$ is a ""bisimulation"" if $R \subseteq \sf{Rel}(\funF)_{\xi}(R)$. ""Bisimilarity"" is the greatest fixpoint $\bisim_\xi$ of $\sf{Rel}(\funF)_{\xi}$ on the lattice of relations. The coinduction principle states that in order to prove $x \bisim_\xi y$, it is sufficient to produce a bisimulation $R$ that contains the pair $(x,y)$. However, bisimulations may be hard to compute or even too large to be computed. In the last years, a great interest has been given to the theory of up-to techniques. Such techniques allow to prove bisimilarity of two states $x \bisim_\xi y$ by exhibiting a relation $R$ such that $(x,y) \in R$ and $R$ is a bisimulation \emph{up-to} some other operator on the lattice $\funP(X \times X)$. More precisely, let $\sf{Tch} : \funP(X \times X) \to \funP(X \times X)$ be a monotone function called ""up-to technique"". The technique $\sf{Tch}$ is ""sound"" (with respect to $\xi$) if for all $R \subseteq X \times X$,
\begin{equation}
R \subseteq (\sf{Rel}(\funF)_{\xi} \circ \sf{Tch})(R) \Rightarrow R \subseteq \bisim_\xi
\end{equation}
"Sound" techniques are not closed under composition, which led the community to introduce the subclass of "compatible" techniques \cite{pous2007complete}. The function $\sf{Tch}$ is $\sf{Rel}(\funF)_{\xi}$-""compatible"" if $\sf{Tch} \circ \sf{Rel}(\funF)_{\xi} \subseteq \sf{Rel}(\funF)_{\xi} \circ \sf{Tch}$. "Compatibility" entails "soundness" and enjoys very good compositional properties. A standard approach to derive "compatible" up-to techniques is to identify a so-called "bialgebra". We can e.g. use Theorem~$4$ of \cite{rot2017enhanced}:

\begin{proposition} \label{prop:bialgebra}
Let $\funF, \funT$ be $\Set$-functors, $\lambda : \funT\funF \Rightarrow \funF\funT$ be a natural transformation, $(X,\alpha)$ a $\funT$-algebra and $(X,\xi)$ an $\funF$-coalgebra. 
Contextual closure $\sf{Ctx}_\alpha$ is defined by
\begin{equation}
\sf{Ctx}_\alpha(R) = \{(\alpha \circ \funT\pi_1(t), \alpha \circ \funT\pi_2(t)) \mid t \in \funT R\}
\end{equation}
Congruence closure is the coupling of equivalence closure and context closure, formally:
\begin{equation}
\sf{Cgr}_\alpha(R) = \bigcup_{n \geq 0} \left( \sf{Tra} \cup \sf{Sym} \cup \sf{Ctx}_\alpha \cup \sf{Rfl} \right)^n
\end{equation}
If $\funF \alpha \circ \lambda_X \circ \funT\xi = \xi \circ \alpha$, meaning that $(X,\alpha,\xi)$ is a $\lambda$-""bialgebra"", $\sf{Ctx}_\alpha$ is $\sf{Rel}(\funF)_\xi$-compatible. 
Under the additional assumption that $\funF$ preserves weak pullbacks, $\sf{Cgr}_\alpha$ is $\sf{Rel}(\funF)_\xi$-compatible. 
\end{proposition}

We place ourselves the situation of Proposition~\ref{prop:gendet}: $\monS, \monT$ are monads on $\Set$, $\funG$ is a $\Set$-functor and there are two "weak distributive laws" $\wdld : \monT \monS \Rightarrow \monS \monT$, $\wdls : \monT \funG \Rightarrow \funG \monT$. Let $(X,c)$ be a $\funG\funS\funT$-coalgebra. In accordance with Lemma~\ref{lem:same}, let the determinized $c$ be $c^+ = \funG \funS \muTX \circ \funG \wdld_{\funT X} \circ \wdls_{\funS\funT X} \circ \funT c : \funT X \to \funG\funS\funT X$. In the case of "distributive laws", the proof that $c^+$ is part of a "bialgebra" would use only multiplication diagrams. Hence, the result still holds in the weak case:

\begin{lemma} \label{lem:ourbialgebra}
The triple $(\funT X, \muTX, c^+)$ is a $\funG\wdld \circ \wdls\funS$-bialgebra.
\end{lemma}


Proposition~\ref{prop:bialgebra} and Lemma~\ref{lem:ourbialgebra} together yield

\begin{theorem} \label{theo:ourcompatible}
\begin{itemize}
\item Contextual closure $\sf{Ctx}_{\muTX}$ is $\sf{Rel}(\funS)_{c^+}$-compatible.
\item If $\funG\funS$ preserves weak pullbacks, then $\sf{Cgr}_{\muTX}$ is $\sf{Rel}(\funS)_{c^+}$-compatible.
\end{itemize}
\end{theorem}

Note that Proposition~\ref{prop:wcanonical}, in which $\monS = \monP$, plays a prominent role in obtaining concrete "weak distributive laws". Indeed, all three Examples~\ref{ex:pb},~\ref{ex:pp} and~\ref{ex:pd} are instances of canonical laws. Thus, it is interesting to note what happens to $\sf{Rel}(\funS)_{c^+}$ if we replace $\monS$ with $\monP$. Straightforward computations show that we retrieve the classical Milner-Park bisimulations of transition systems:
\begin{align}
\sf{Rel}(\funP)_{c^+}(R) =  \{(u,v) & \in \funT X \times \funT X \text{ such that } \\
& \forall u' \in c^+(u), \exists v' \in c^+(v) \text{ such that } (u',v') \in R \\
& \forall v' \in c^+(\phi), \exists u' \in c^+(u) \text{ such that } (u',v') \in R \}
\end{align}

\subsection{Application to Alternating Automata} \label{subsec:appaa}

Consider the concrete case of "alternating automata" with the data presented in Section~\ref{subsec:detaa}. We call the corresponding context closure ""union closure"" because of the expression:
\begin{equation}
\sf{Ctx}_{\muPX}(R) = \left\{ \left(\bigcup_{i\in I} A_i,\bigcup_{i\in I} B_i\right) \mid \forall i \in I, (A_i,B_i) \in R \right\}
\end{equation}

Both $\funM$ and $\funP$ preserve weak pullbacks, see e.g. Proposition~$4.2.6$ in \cite{jacobs2017introduction}. Using Theorem~\ref{theo:ourcompatible}, we get that

\begin{proposition} \label{prop:aacompat}
"Union closure" and its corresponding congruence are "compatible" with respect to "alternating automata" determinized once.
\end{proposition}

\begin{example}
Consider the "alternating automaton" $c_0$ of Figure~\ref{fig:aa} and its determinized $c_0^+$ pictured in Figure~\ref{fig:aadet}. Following the steps of~\cite{Bon17}, we consider various bisimulations on $c_0^+$. First, one can see that $\{x_2\} \bisim_{c_0^+} \{y_2\}$ because $(\{x_2\},\{y_2\})$ is a "bisimulation". Indeed, there is exactly one arrow going out of these states, and this is a loop. Similarly, $(\{x_3\},\{y_3\})$ is a "bisimulation" because there is no arrow going out of these states. Now let us try to prove that $\{x_0\} \bisim_{c_0^+} \{y_0\}$. By exploring every possible transition, one can easily see that a bisimulation relating $\{x_0\}$ and $\{y_0\}$ of minimal cardinality is
\begin{align}
R = \{ & (\{x_0\},\{y_0\}), (\{x_1\},\{y_1,y_2\}), (\{x_3\},\{y_3\}), (\{x_1,x_3\},\{y_1,y_2,y_3\}), \\
& (\{x_1,x_2\},\{y_1,y_2\}),(\{x_2,x_3\},\{y_2,y_3\}),(\{x_1,x_2,x_3\},\{y_1,y_2,y_3\})  \}
\end{align}
\end{example}
However, the following smaller relation is a "bisimulation" up-to congruence witnessing the fact that $\{x_0\} \bisim_{c_0^+} \{y_0\}$:
\begin{equation}
R_0 = \{ (\{x_0\},\{y_0\}),(\{x_1\},\{y_1,y_2\}),(\{x_2\},\{y_2\}),(\{x_3\},\{y_3\}) \}
\end{equation}
The only non-trivial verifications concern the pair $(\{x_1\},\{y_1,y_2\})$:
\begin{center}
\begin{tikzcd}
\{x_1\} \arrow[r, "a"] \arrow[dd, "R_0" description, no head, dashed] & {\{x_1,x_2\}} \arrow[dd, "\sf{Cgr}_{c_0^+}(R_0)" description, no head, dashed] & \{x_1\} \arrow[r, "a"] \arrow[dd, "R_0" description, no head, dashed] & {\{x_2,x_3\}} \arrow[dd, "\sf{Cgr}_{c_0^+}(R_0)" description, no head, dashed] & \{x_1\} \arrow[r, "a"] \arrow[dd, "R_0" description, no head, dashed] & {\{x_1,x_2,x_3\}} \arrow[dd, "\sf{Cgr}_{c_0^+}(R_0)" description, no head, dashed] \\
                                                                      &                                                                                &                                                                       &                                                                                &                                                                       &                                                                                    \\
{\{y_1,y_2\}} \arrow[r, "a"]                                          & {\{y_1,y_2\}}                                                                  & {\{y_1,y_2\}} \arrow[r, "a"]                                          & {\{y_2,y_3\}}                                                                  & {\{y_1,y_2\}} \arrow[r, "a"]                                          & {\{y_1,y_2,y_3\}}                                                                 
\end{tikzcd}
\end{center}


\begin{remark}
For finite systems, the congruence obtained from "union closure" is the same thing as the up-to technique presented in~\cite{bonchi2013checking}. Note however that in \emph{loc. cit.} the authors use compatibility of congruence closure to compute bisimulations on deterministic automata i.e. $\funM$-coalgebras, whereas we use it to compute Milner-Park bisimulations on non-deterministic automata i.e. $\funM\funP$-coalgebras. In particular, in the setting of \emph{loc. cit.}, bisimilarity and behavioural equivalence coincide because there is a final $\funM$-coalgebra --- whereas Milner-Park bisimilarity strictly entails behavioural equivalence on $\funM\funP$-coalgebras.
\end{remark}

\subsection{Application to Probabilistic Automata} \label{subsec:apppa}

Let the functor $-^A : \Set \to \Set$ be the second component of the "machine functor" $\funM$. Consider the "weak distributive law" $\wdld : \monD \monP \Rightarrow \monP \monD$ of Example~\ref{ex:pd} along with the following "distributive law" $\wdls : \monD(-^A) \Rightarrow \monD(-)^A$:
\begin{equation}
\wdls_X(\Phi)(a)(x) = \sum_{f \in X^A, f(a) = x} \Phi(f)
\end{equation}
Note that given a "weak distributive law" of type $\monT \monS \Rightarrow \monS \monT$ and using the terminology of \cite{Bon17}, $\funSwl$ is a $(\monST,\funS)$-quasi-lax lifting because of the equality $\monST = \funUuT \funSwl \funFuT$ and the monic $\iota : \funUuT \funSwl \Rightarrow \funS \funUuT$. Using this observation, and going through the construction of Section~$6$ in~\cite{Bon17}, one can see that it amounts exactly to the left half of the generalized determinization of our Proposition~\ref{prop:gendet}.
Hence, determinizing a probabilistic automaton $c : X \to (\funP \funD X)^A$ into a $c^+ : \funD X \to (\funP \funD X)^A$ yields the same belief-state transformer as in~\cite{Bon17}. As this fact was already remarked in~\cite{goypetrisan} --- modulo the set of labels $A$ --- we are not going to detail this more. Back to up-to techniques, note that context closure with respect to the free $\monD$-algebra is ""convex hull"":
\begin{align}
\sf{Ctx}_{\muDX}(R) & = \left\{\left( \sum_{i \in I} p_i \phi_i, \sum_{i\in I} p_i \psi_i \right) \mid I \text{ finite }, \sum_{i\in I} p_i = 1, (\phi_i,\psi_i) \in R, p_i \in (0,1] \right\}
\end{align}
Both $(-)^A$ and $\funP$ preserve weak pullbacks~\cite{jacobs2017introduction}, hence Theorem~\ref{theo:ourcompatible} yields a result already stated in~\cite{Bon17}:

\begin{proposition} \label{prop:pacompat}
"Convex hull" and its corresponding congruence are "compatible" with respect to the belief-state transformer determinization of a probabilistic automaton.
\end{proposition}

\begin{remark} \label{rem:wdlmorphism}
As previously mentioned, the structure of the alternating automaton of Figure~\ref{fig:aa} is the same as the one of the probabilistic automaton example in~\cite{Bon17} --- and determinization of both systems remain strikingly similar. This is due to the fact that the natural transformation $\supp : \funD \Rightarrow \funP$ is a morphism of "weak distributive laws", in the sense that the following diagram commutes, where vertical arrows are the canonical "weak distributive laws" described in Example~\ref{ex:pd} and Example~\ref{ex:pp}.
\begin{center}
\begin{tikzcd}
\funD \funP \ar[r,"\sf{supp}\funP" above] \ar[d] & \funP \funP \ar[d] \\
\funP \funD \ar[r,"\funP \sf{supp}" below] & \funP \funP 
\end{tikzcd}
\end{center}
\end{remark}

\section{Conclusion}

\subsection*{Summary of Results}

The slogan of this article might be: if there is almost a "distributive law", then everything goes almost smoothly. 
By paying the light price of removing one unit diagram, or equivalently dislocating one identity natural transformation into two that are pseudo-inverses, one retrieves a lot of the results that have drawn interest from the community over the last few decades.

First, we noticed that suppressing one of the monad structures leads to a notion of "weak distributive law" of a monad over a functor. As "weak distributive laws" can give rise to new monads, we provided an abstract compositionality result (Theorem~\ref{theo:compo}) to derive a "weak distributive law" of the composite monad. We also fully extended generalized determinization to the setting of "weak distributive laws" (Proposition~\ref{prop:gendet}) --- this is the logical continuation of~\cite{goypetrisan} where only one (monad-monad) "weak distributive law" was involved. Finally, we identified that the "bialgebraic" approach to "up-to techniques" could be adapted effortlessly to the weak framework (Theorem~\ref{theo:ourcompatible}), hence providing a full batch of sound "up-to techniques" that were out of the reach of plain "distributive laws".

Our two primary examples, "alternating automata" and probabilistic automata, come directly from the two black sheep that have caused much ink to flow in the recent years: powerset over powerset, and distribution over powerset. Our framework explains how generalized determinization of "alternating automata" modelled with double "covariant powerset monad" occurs in a "distributive law"-like manner. To our knowledge, this is the first time that this is performed. This determinization is sound in the sense that semantics is preserved, and canonical in the sense that it corresponds to the canonical powerset-powerset "weak distributive law". Our results also yield that union closure is "compatible" with Milner-Park "bisimulation" for "alternating automata", and we retrieve back that "convex hull" is "compatible" with Milner-Park "bisimulation" for probabilistic automata.

\subsection*{Future Work}

Exciting insights are lying ahead. Can we tell more about the "weak distributive law" morphisms of Remark~\ref{rem:wdlmorphism}, e.g. can they be used as a source of new examples of "weak distributive laws"? Such morphisms have already been studied in the strong framework e.g. in~\cite{power2002combining},~\cite{bonsangue2013presenting}. This might be further investigated to provide answers to the more general question: is there something more beyond, once the case of $\monD\monP \Rightarrow \monP\monD$ and $\monP\monP \Rightarrow \monP\monP$ is settled? Our framework introduces monad-functor "weak distributive laws", but what are instances of this that do not arise directly from forgetting a monad structure?

There is a variety of other ways of performing generalized determinization with composed "weak distributive laws". One example in the case when "Yang-Baxter" holds is to start from a system in $\Coalg{\funG\funS\funT}$, use $\piwl^*$ to go into $\Coalg{\funG\monST}$, and then perform determinization with respect to the $\monST\funG \Rightarrow \funG\monST$ obtained in Theorem~\ref{theo:compo}. For the moment it is unclear how this construction would relate to the double determinization mentioned in Remark~\ref{rem:chain}. A further point is we did not prove that replacing the functor with a third monad in the compositionality theorem makes diagram $(\muR)$ with respect to $\monST$ fail. This would be an interesting result, as this would entail that in Theorem~\ref{theo:compo} one can not add further monads as in~\cite{cheng2011iterated}.



\bibliography{biblio}

\begin{thebibliography}{10}

\bibitem{barr1970relational}
Michael Barr.
\newblock Relational algebras.
\newblock In {\em Reports of the Midwest Category Seminar IV}, pages 39--55.
  Springer, 1970.

\bibitem{beck:distributiveLaws}
Jon Beck.
\newblock Distributive laws.
\newblock In B.~Eckmann, editor, {\em Seminar on Triples and Categorical
  Homology Theory}, pages 119--140, Berlin, Heidelberg, 1969. Springer Berlin
  Heidelberg.

\bibitem{bertrand2018coalgebraic}
Meven Bertrand and Jurriaan Rot.
\newblock Coalgebraic determinization of alternating automata.
\newblock {\em arXiv preprint arXiv:1804.02546}, 2018.

\bibitem{Bohm2010}
Gabriella Böhm.
\newblock The weak theory of monads.
\newblock {\em Advances in Mathematics}, 225(1):1 -- 32, 2010.
\newblock URL:
  \url{http://www.sciencedirect.com/science/article/pii/S0001870810000794},
  \href {https://doi.org/https://doi.org/10.1016/j.aim.2010.02.015}
  {\path{doi:https://doi.org/10.1016/j.aim.2010.02.015}}.

\bibitem{DBLP:conf/csl/BonchiPPR14}
Filippo Bonchi, Daniela Petri{\c s}an, Damien Pous, and Jurriaan Rot.
\newblock Coinduction up-to in a fibrational setting.
\newblock In {\em {CSL-LICS}}, pages 20:1--20:9. {ACM}, 2014.

\bibitem{DBLP:journals/acta/BonchiPPR17}
Filippo Bonchi, Daniela Petri{\c s}an, Damien Pous, and Jurriaan Rot.
\newblock A general account of coinduction up-to.
\newblock {\em Acta Inf.}, 54(2):127--190, 2017.

\bibitem{bonchi2013checking}
Filippo Bonchi and Damien Pous.
\newblock Checking nfa equivalence with bisimulations up to congruence.
\newblock {\em ACM SIGPLAN Notices}, 48(1):457--468, 2013.

\bibitem{Bon17}
Filippo Bonchi, Alexandra Silva, and Ana Sokolova.
\newblock The power of convex algebras.
\newblock In {\em 28th International Conference on Concurrency Theory, {CONCUR}
  2017, September 5-8, 2017, Berlin, Germany}, pages 23:1--23:18, 2017.
\newblock \href {https://doi.org/10.4230/LIPIcs.CONCUR.2017.23}
  {\path{doi:10.4230/LIPIcs.CONCUR.2017.23}}.

\bibitem{Bon19}
Filippo Bonchi, Ana Sokolova, and Valeria Vignudelli.
\newblock The theory of traces for systems with nondeterminism and probability.
\newblock In {\em 34th Annual {ACM/IEEE} Symposium on Logic in Computer
  Science, {LICS} 2019, Vancouver, BC, Canada, June 24-27, 2019}, pages 1--14,
  2019.
\newblock \href {https://doi.org/10.1109/LICS.2019.8785673}
  {\path{doi:10.1109/LICS.2019.8785673}}.

\bibitem{bonsangue2013presenting}
Marcello~M Bonsangue, Helle~Hvid Hansen, Alexander Kurz, and Jurriaan Rot.
\newblock Presenting distributive laws.
\newblock In {\em International Conference on Algebra and Coalgebra in Computer
  Science}, pages 95--109. Springer, 2013.

\bibitem{chandra1981lj}
Ashok~K. Chandra, Dexter~C. Kozen, and Larry~J. Stockmeyer.
\newblock Alternation.
\newblock {\em J. ACM}, 28(1):114–133, January 1981.
\newblock \href {https://doi.org/10.1145/322234.322243}
  {\path{doi:10.1145/322234.322243}}.

\bibitem{cheng2011iterated}
Eugenia Cheng.
\newblock Iterated distributive laws.
\newblock In {\em Mathematical Proceedings of the Cambridge Philosophical
  Society}, volume 150, pages 459--487. Cambridge University Press, 2011.

\bibitem{Gar19}
Richard Garner.
\newblock The {V}ietoris monad and weak distributive laws.
\newblock {\em Applied Categorical Structures}, Oct 2019.
\newblock \href {https://doi.org/10.1007/s10485-019-09582-w}
  {\path{doi:10.1007/s10485-019-09582-w}}.

\bibitem{goypetrisan}
Alexandre Goy and Daniela Petrisan.
\newblock Combining probabilistic and non-deterministic choice via weak
  distributive laws.
\newblock 2020.

\bibitem{jacobs2006bialgebraic-review}
Bart Jacobs.
\newblock {\em A Bialgebraic Review of Deterministic Automata, Regular
  Expressions and Languages}, pages 375--404.
\newblock Springer Berlin Heidelberg, Berlin, Heidelberg, 2006.
\newblock \href {https://doi.org/10.1007/11780274_20}
  {\path{doi:10.1007/11780274_20}}.

\bibitem{jacobs2017introduction}
Bart Jacobs.
\newblock {\em Introduction to Coalgebra}, volume~59.
\newblock Cambridge University Press, 2017.

\bibitem{klin2015coalgebraic}
Bartek Klin and Jurriaan Rot.
\newblock Coalgebraic trace semantics via forgetful logics.
\newblock In {\em International Conference on Foundations of Software Science
  and Computation Structures}, pages 151--166. Springer, 2015.

\bibitem{KLIN2018261}
Bartek Klin and Julian Salamanca.
\newblock Iterated covariant powerset is not a monad.
\newblock {\em Electronic Notes in Theoretical Computer Science}, 341:261 --
  276, 2018.
\newblock Proceedings of the Thirty-Fourth Conference on the Mathematical
  Foundations of Programming Semantics (MFPS XXXIV).
\newblock URL:
  \url{http://www.sciencedirect.com/science/article/pii/S157106611830094X},
  \href {https://doi.org/https://doi.org/10.1016/j.entcs.2018.11.013}
  {\path{doi:https://doi.org/10.1016/j.entcs.2018.11.013}}.

\bibitem{klin2018iterated}
Bartek Klin and Julian Salamanca.
\newblock Iterated covariant powerset is not a monad.
\newblock {\em Electr. Notes Theor. Comput. Sci.}, 341:261--276, 2018.

\bibitem{pous2007complete}
Damien Pous.
\newblock Complete lattices and up-to techniques.
\newblock In {\em Asian Symposium on Programming Languages and Systems}, pages
  351--366. Springer, 2007.

\bibitem{pous_sangiorgi_2011}
Damien Pous and Davide Sangiorgi.
\newblock {\em Enhancements of the bisimulation proof method}, page 233–289.
\newblock Cambridge Tracts in Theoretical Computer Science. Cambridge
  University Press, 2011.
\newblock \href {https://doi.org/10.1017/CBO9780511792588.007}
  {\path{doi:10.1017/CBO9780511792588.007}}.

\bibitem{power2002combining}
John Power and Hiroshi Watanabe.
\newblock Combining a monad and a comonad.
\newblock {\em Theoretical Computer Science}, 280(1-2):137--162, 2002.

\bibitem{rot2017enhanced}
Jurriaan Rot, Filippo Bonchi, Marcello Bonsangue, Damien Pous, Jan Rutten, and
  Alexandra Silva.
\newblock Enhanced coalgebraic bisimulation.
\newblock {\em Mathematical Structures in Computer Science}, 27(7):1236--1264,
  2017.

\bibitem{SilvaBBR10}
Alexandra Silva, Filippo Bonchi, Marcello~M. Bonsangue, and Jan J. M.~M.
  Rutten.
\newblock Generalizing the powerset construction, coalgebraically.
\newblock In {\em {FSTTCS}}, volume~8 of {\em LIPIcs}, pages 272--283. Schloss
  Dagstuhl - Leibniz-Zentrum fuer Informatik, 2010.

\bibitem{Street_weak_laws}
Ross Street.
\newblock Weak distributive laws.
\newblock {\em Theory and Applications of Categories}, 22(12):313--320, 2009.

\bibitem{street2009weak}
Ross Street et~al.
\newblock Weak distributive laws.
\newblock 2009.

\bibitem{Varacca03PhD}
Daniele Varacca.
\newblock Probability, nondeterminism and concurrency: Two denotational models
  for probabilistic computation.
\newblock Technical report, PhD thesis, Univ. Aarhus, 2003. BRICS Dissertation
  Series, 2003.

\bibitem{Zwart_Marsden}
M.~Zwart and D.~Marsden.
\newblock No-go theorems for distributive laws.
\newblock In {\em 2019 34th Annual ACM/IEEE Symposium on Logic in Computer
  Science (LICS)}, pages 1--13, Los Alamitos, CA, USA, jun 2019. IEEE Computer
  Society.
\newblock URL:
  \url{https://doi.ieeecomputersociety.org/10.1109/LICS.2019.8785707}, \href
  {https://doi.org/10.1109/LICS.2019.8785707}
  {\path{doi:10.1109/LICS.2019.8785707}}.

\end{thebibliography}

\newpage
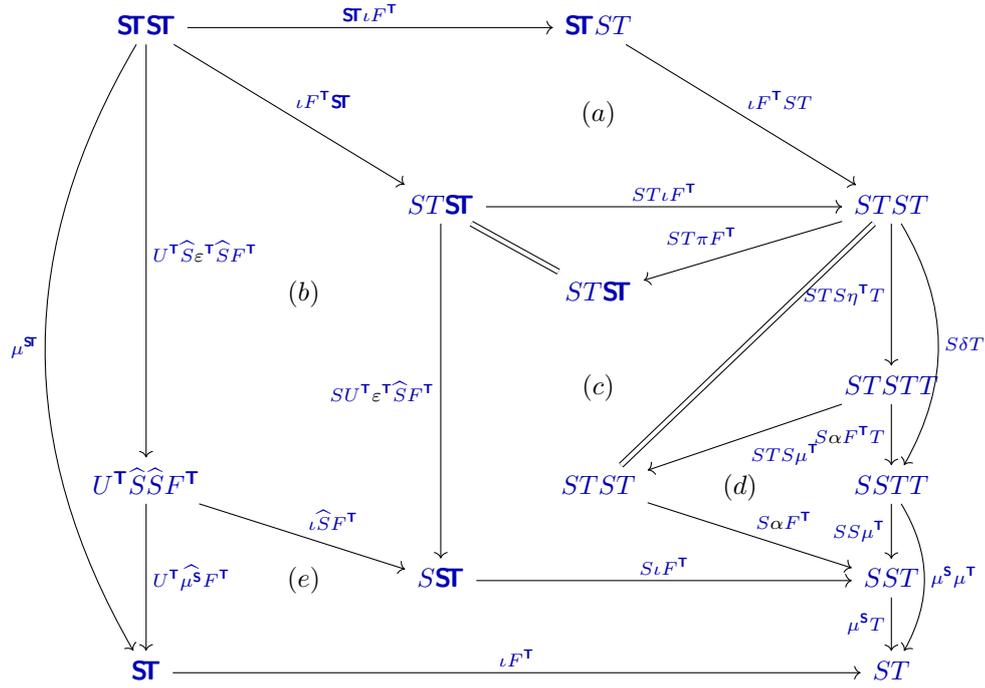
\appendix

\section{Proofs of Section~\ref{sec:weakframework}}
 
\subsection{Proof of Theorem~\ref{theo:compo}}

Let us prove Theorem~\ref{theo:compo} i.e. $\wdll = \funG \piwl^* \circ \wdlt \funT \circ \funS \wdls \circ \iotawl^* \funG$ is a "weak distributive law" of type $\monST \funG \Rightarrow \funG \monST$. Recall that notation $\piwl^*$ (resp. $\iotawl^*$) stands for $\piwl \funFuT$ (resp. $\iotawl \funFuT$).\\

The proof that the $(\muST)$ diagram commutes begins as follows:

\adjustbox{scale=0.9,center}{%
\begin{tikzcd}
\monST\monST \funG \arrow[r, "\monST \iotawl^* \funG"'] \arrow[dddddd, "\muST\funG"'] \arrow[rrrr, "\monST \wdll", bend left] & \monST\funS\funT\funG \arrow[r, "\monST \funS \wdls"'] \arrow[rddd, "\iotawl^* \funS\funT\funG"'] & \monST\funS\funG\funT \arrow[r, "\monST \wdlt \funT"']                                            & \monST\funG\funS\funT \arrow[r, "\monST\funG \piwl^*"'] \arrow[rddd, "\iotawl^* \funG\funS\funT"'] & \monST\funG\monST \arrow[r, "\iotawl^* \funG\monST"'] \arrow[rrrr, "\wdll \monST", bend left]         & \funS\funT\funG\monST \arrow[r, "\funS \wdls \monST"']       & \funS\funG\funT\monST \arrow[r, "\wdlt\funT \monST"']                                                  & \funG\funS\funT\monST \arrow[r, "\funG \piwl^* \monST"'] & \funG\monST\monST \arrow[dddddd, "\funG\muST"] \\
                                                                                                                              &                                                                                                   & (1)                                                                                               &                                                                                                    & (2)                                                                                                   &                                                              & (3)                                                                                                    &                                                          &                                                \\
                                                                                                                              &                                                                                                   &                                                                                                   &                                                                                                    &                                                                                                       &                                                              &                                                                                                        &                                                          &                                                \\
                                                                                                                              &                                                                                                   & \funS\funT\funS\funT\funG \arrow[r, "\funS\funT\funS\wdls"'] \arrow[dd, "\funS \wdld \funT\funG"] & \funS\funT\funS\funG\funT \arrow[r, "\funS\funT\wdlt\funT"']                                       & \funS\funT\funG\funS\funT \arrow[r, "\funS\wdls\funS\funT"'] \arrow[ruuu, "\funS\funT\funG \piwl^*"'] & \funS\funG\funT\funS\funT \arrow[r, "\wdlt\funT\funS\funT"'] & \funG\funS\funT\funS\funT \arrow[dd, "\funG\funS\wdld \funT"'] \arrow[ruuu, "\funG\funS\funT\piwl^*"'] &                                                          &                                                \\
                                                                                                                              & (4)                                                                                               &                                                                                                   &                                                                                                    & (5)                                                                                                   &                                                              &                                                                                                        & (6)                                                      &                                                \\
                                                                                                                              &                                                                                                   & \funS\funS\funT\funT\funG \arrow[d, "\muS\muT\funG"]                                              &                                                                                                    &                                                                                                       &                                                              & \funG\funS\funS\funT\funT \arrow[d, "\funG\muS\muT"']                                                  &                                                          &                                                \\
\monST\funG \arrow[rr, "\iotawl^* \funG"] \arrow[rrrrrrrr, "\wdll", bend right]                                               &                                                                                                   & \funS\funT\funG \arrow[rr, "\funS \wdls"]                                                         &                                                                                                    & \funS\funG\funT \arrow[rr, "\wdlt \funT"]                                                             &                                                              & \funG\funS\funT \arrow[rr, "\funG \piwl^*"]                                                            &                                                          & \funG\monST                                   
\end{tikzcd}
}
The three external curved diagrams commute by definition of $\wdll$. Diagram $(1)$,$(2)$,$(3)$ commute respectively because of naturality of $\iotawl^*$, naturality of $\iotawl^*$, naturality of $\wdlt\funT \circ \funS\wdls$. Diagram $(5)$ commutes as in \cite{cheng2011iterated} by making use of the "Yang-Baxter" hypothesis. Hence only diagrams $(4)$ and $(6)$ remain to be proved commutative. Note that as $\monST$ is obtained from the composite adjunction
\begin{tikzcd}
\EMSwl \ar[r,yshift=-0.4em] & \EMT \ar[l,yshift=0.6em, "\bot" below] \ar[r, yshift = -0.4em] & \catC \ar[l, yshift = 0.6em, "\bot" below]
\end{tikzcd}
, we have $\muST = \funUuT \muSwl \funFuT \circ \funUuT \funSwl \varepsilon^{\monT} \funSwl \funFuT$ where $\varepsilon^{\monT}$ is the counit of adjunction $\funFuT \dashv \funUuT$.
Recall also that according to the proof of Proposition \ref{prop:weakbij}, $\wdld = \begin{tikzcd}
\funT\funS \ar[r,"\funT\funS\etaT" above] & \funT\funS\funT \ar[r, "\funT \piwl\funFuT" above] & \funT \funUuT \funSwl \funFuT \ar[r, "\funUuT \varepsilon^{\monT} \funSwl \funFuT" above] & \funUuT \funSwl \funFuT \ar[r, "\iotawl\funFuT" above] & \funS\funT \end{tikzcd}$. Denote $\alpha = \iotawl \circ \funUuT \varepsilon^{\monT} \funSwl \circ \funT \piwl : \funT \funS \funUuT \Rightarrow \funS \funUuT$ for convenience.

Let us first prove that diagram $(4)$ commutes. It suffices to prove that the following diagram commutes (we just suppressed all the $\funG$ on the right):

\begin{center}
\begin{tikzcd}
\monST\monST \arrow[ddddd, "\funUuT \funSwl \varepsilon^{\monT} \funSwl \funFuT"] \arrow[rrr, "\monST \iotawl \funFuT"] \arrow[ddddddd, "\muST"', bend right] \arrow[rrdd, "\iotawl \funFuT \monST"] &     &                                                                                                                     & \monST\funS\funT \arrow[rrdd, "\iotawl \funFuT \funS\funT"] &     &                                                                                                                                                                                 \\
                                                                                                                                                                                             &     &                                                                                                                     & (a)                                                         &     &                                                                                                                                                                                 \\
                                                                                                                                                                                             &     & \funS\funT\monST \arrow[dddd, "\funS \funUuT \varepsilon^{\monT} \funSwl \funFuT"'] \arrow[rrr, "\funS\funT\iotawl\funFuT"] &                                                             &     & \funS\funT\funS\funT \arrow[ddd, "\funS \wdld \funT", bend left] \arrow[dd, "\funS\funT\funS \etaT \funT"'] \arrow[llddd, equal] \arrow[lld, "\funS \funT \piwl \funFuT"'] \\
                                                                                                                                                                                             & (b) &                                                                                                                     & \funS\funT\monST \arrow[lu, equal]                     &     &                                                                                                                                                                                 \\
                                                                                                                                                                                             &     &                                                                                                                     & (c)                                                         &     & \funS\funT\funS\funT\funT \arrow[d, "\funS \alpha \funFuT\funT"'] \arrow[lld, "\funS\funT\funS\muT"]                                                                            \\
\funUuT \funSwl \funSwl \funFuT \arrow[dd, "\funUuT \muSwl \funFuT"] \arrow[rrd, "\iotawl \funSwl \funFuT"]                                                                                  &     &                                                                                                                     & \funS\funT\funS\funT \arrow[rrd, "\funS\alpha \funFuT"]     & (d) & \funS\funS\funT\funT \arrow[dd, "\muS\muT", bend left] \arrow[d, "\funS\funS\muT"']                                                                                             \\
                                                                                                                                                                                             & (e) & \funS\monST \arrow[rrr, "\funS \iotawl \funFuT"]                                                                    &                                                             &     & \funS\funS\funT \arrow[d, "\muS\funT"']                                                                                                                                         \\
\monST \arrow[rrrrr, "\iotawl \funFuT"]                                                                                                                                                      &     &                                                                                                                     &                                                             &     & \funS\funT                                                                                                                                                                     
\end{tikzcd}
\end{center}

The three external triangles with bent arrows commute by definition. The two internal triangles commute because of a monad property and the fact $\pi \circ \iota = 1$. Diagrams $(a),(b)$ commute by naturality of $\iotawl$. Diagram $(c)$ commutes by definition of $\alpha$. Diagram $(e)$ commute because of the $(\iota\mu)$ diagram of "weak liftings". The diagram $(d)$ is a little bit harder. By monad properties, $\muTX : \funT\funT X \Rightarrow \funT X$ is a actually a morphism $[\muTX] : (\funT\funT X, \muTTX) \to (\funT X, \muTX)$ in $\EMT$. This yield a natural transformation $[\muT] : \funFuT \funT \Rightarrow \funFuT$ such that $\funUuT [\muT] = \muT$. Diagram $(d)$ then commutes by replacing $\muT$ with $\funUuT [\muT]$ and using naturality of $\alpha : \funT\funS\funUuT \Rightarrow \funS\funUuT$.\\

Let us finally prove that diagram $(6)$ commutes. It amounts to the following:

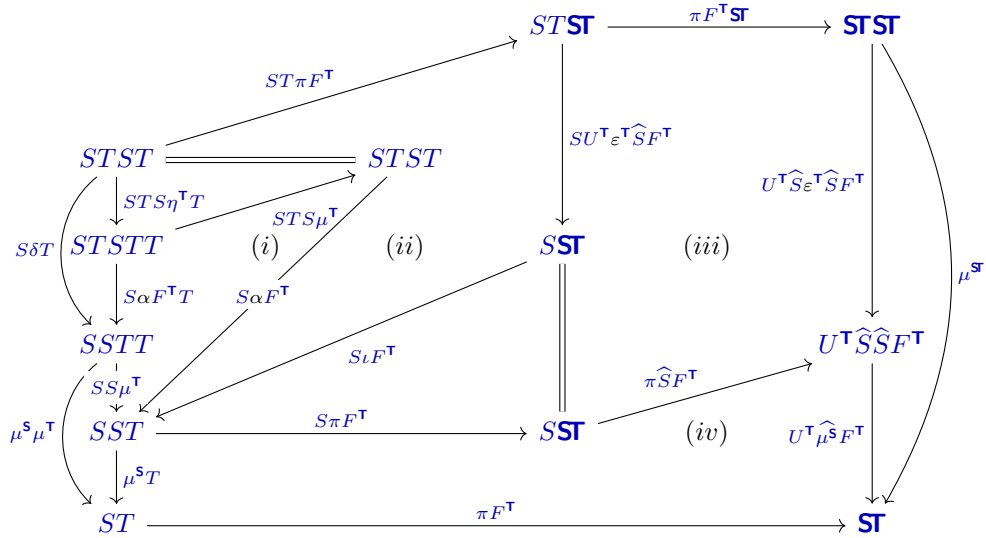
\begin{center}
\begin{tikzcd}
                                                                                                                                                                         &     &                                                                      & \funS\funT\monST \arrow[rr, "\piwl \funFuT \monST"] \arrow[ddd, "\funS \funUuT \varepsilon^{\monT} \funSwl \funFuT"] &       & \monST\monST \arrow[dddddd, "\muST", bend left] \arrow[dddd, "\funUuT \funSwl \varepsilon^{\monT} \funSwl \funFuT"'] \\
                                                                                                                                                                         &     &                                                                      &                                                                                                              &       &                                                                                                              \\
\funS\funT\funS\funT \arrow[dd, "\funS\wdld\funT"', bend right=49] \arrow[rrruu, "\funS\funT\piwl\funFuT"] \arrow[d, "\funS\funT\funS\etaT\funT"] \arrow[rr, equal] &     & \funS\funT\funS\funT \arrow[llddd, "\funS\alpha\funFuT" description] &                                                                                                              &       &                                                                                                              \\
\funS\funT\funS\funT\funT \arrow[d, "\funS\alpha\funFuT \funT"] \arrow[rru, "\funS\funT\funS\muT"']                                                                      & (i) & (ii)                                                                 & \funS\monST \arrow[dd, equal] \arrow[llldd, "\funS\iotawl\funFuT"]                                      & (iii) &                                                                                                              \\
\funS\funS\funT\funT \arrow[dd, "\muS\muT"', bend right=49] \arrow[d, "\funS\funS\muT" description]                                                                      &     &                                                                      &                                                                                                              &       & \funUuT \funSwl \funSwl \funFuT \arrow[dd, "\funUuT \muSwl \funFuT"']                                        \\
\funS\funS\funT \arrow[d, "\muS\funT"] \arrow[rrr, "\funS\piwl\funFuT"]                                                                                                  &     &                                                                      & \funS\monST \arrow[rru, "\piwl\funSwl\funFuT"]                                                               & (iv)  &                                                                                                              \\
\funS\funT \arrow[rrrrr, "\piwl\funFuT"]                                                                                                                                 &     &                                                                      &                                                                                                              &       & \monST                                                                                                      
\end{tikzcd}
\end{center}

The three external triangles with bent arrows commute by definition. The two internal triangles commute because of a monad property and the fact $\pi \circ \iota = 1$. Diagram $(i)$ commutes for the same reasons as for diagram $(d)$ in the preceding diagram. Diagram $(ii)$ commutes by definition of $\alpha$, diagram $(iii)$ by naturality of $\pi$, and diagram $(iv)$ by using the $(\pi\mu)$ diagram of "weak liftings". This achieves the proof that the $(\muST)$ diagram commutes.\\\\

Now assume that "weak distributive laws" $\wdls$ and $\wdlt$ are actually strong. The following diagram proves that the $(\etaST)$ diagram commutes i.e. that $\wdll$ is strong:

\begin{center}
\begin{tikzcd}
\monST\funG \arrow[rr, "\iotawl^* \funG"] &  & \funS\funT\funG \arrow[rr, "\funS\wdls"]                                                 &  & \funS\funG\funT \arrow[rr, "\wdlt\funT"]                                                                                                                                     &  & \funG\funS\funT \arrow[rr, "\funG\piwl^*"]                                                                          &  & \funG\monST \\
                                          &  &                                                                                          &  &                                                                                                                                                                              &  &                                                                                                                     &  &             \\
                                          &  & \funT\funG \arrow[lluu, "\funUuT \etaSwl \funFuT \funG"'] \arrow[uu, "\etaS\funT\funG"'] &  & \funS\funG \arrow[lluu, "\funS\etaT\funG"'] \arrow[uu, "\funS\funG\etaT"]                                                                                                    &  & \funG\funT \arrow[lluu, "\etaS\funG\funT"] \arrow[uu, "\funG\etaS\funT"] \arrow[rruu, "\funG\funUuT\etaSwl\funFuT"] &  &             \\
                                          &  &                                                                                          &  &                                                                                                                                                                              &  &                                                                                                                     &  &             \\
                                          &  &                                                                                          &  & \funG \arrow[lllluuuu, "\etaST\funG", bend left] \arrow[rrrruuuu, "\funG\etaST"', bend right] \arrow[lluu, "\etaT\funG"] \arrow[uu, "\etaS\funG"] \arrow[rruu, "\funG\etaT"] &  &                                                                                                                     &  &            
\end{tikzcd}
\end{center}
The two external triangles commute by definition of $\etaST$. The leftmost and rightmost triangles commute by "weak liftings" properties $(\pi\eta)$ and $(\iota\eta)$. The two squares commute by naturality. Finally, the two remaining triangles use the $(\etaT)$ and $(\etaS)$ diagrams of "weak distributive laws" with respect to $\funG$.\\
This achieves the proof of Theorem~\ref{theo:compo}.

\subsection{Proof of Example~\ref{ex:compo}}

We prove that the following "Yang-Baxter" diagram related to Example~\ref{ex:compo} commutes. 

\begin{center}
\begin{tikzcd} \label{diag:yangbaxter}
                                                                                 & \funP \funM \funP \arrow[r, "\wdls \funP"] & \funM \funP \funP \arrow[rd, "\funM \wdld"]          &                   \\
\funP \funP \funM \arrow[ru, "\funP \wdlt"] \arrow[rd, "\wdld \funM"'] &                                                       &                                                       & \funM \funP \funP \\
                                                                                 & \funP \funP \funM \arrow[r, "\funP \wdls"']  & \funP \funM \funP \arrow[ru, "\wdlt \funP"'] &                  
\end{tikzcd}
\end{center}

Let us first compute both sides of the equation. Let $\mathcal{S} \in \funP\funP\funM X$.\\ The pair $(\funM\wdld \circ \wdls \funP \circ \funP\wdlt)_X(\mathcal{S})$ has output 
\begin{equation} \bigwedge_{S \in \mathcal{S}} \bigvee_{(o,f) \in S} o \end{equation}
and $a$-transition 
\begin{equation} \{V \subseteq X \mid V \subseteq \{f(a) \mid (o,f) \in S, S \in \mathcal{S}\} \text{ and } \forall S \in \mathcal{S}, \{f(a) \mid (o,f) \in S\} \cap V \neq \emptyset \} \end{equation}
The pair $(\wdlt \funP \circ \funP \wdls \circ \wdld\funM)_X(\mathcal{S}) $ has output 
\begin{equation} \bigvee_{\left(W \subseteq \funM X \text{ s.t. } W \subseteq \bigcup \mathcal{S} \text{ and } \forall S \in \mathcal{S}, S \cap W \neq \emptyset\right)} \bigwedge_{(o,f) \in W} o \end{equation}
and $a$-transition 
\begin{equation} \{\{f(a) \mid (o,f) \in W \} \mid W \subseteq \funM X \text{ s.t. } W \subseteq \bigcup \mathcal{S} \text{ and } \forall S \in \mathcal{S}, S \cap W \neq \emptyset \} \end{equation}

%

Let us first prove that outputs are equal. As they equal either $0$ or $1$, this amounts to showing the following lemma.

\begin{lemma}
The two following propositions are equivalent:
\begin{enumerate}
\item[$(i)$] For all $S \in \mathcal{S}$, there is $f \in X^A$ such that $(1,f) \in S$.
\item[$(ii)$] There is $W \subseteq \funM X$ such that $W \subseteq \bigcup \mathcal{S}$, for all $S \in \mathcal{S}$ one has $S \cap W \neq \emptyset$, and for all $(o,f) \in W$, $o = 1$.
\end{enumerate}
\end{lemma}

\begin{proof}
We begin by the implication $(i) \Rightarrow (ii)$. For every $S \in \mathcal{S}$ we can fix $f_S \in X^A$ such that $(1,f) \in S$. Let $W = \{(1,f_S) \mid S \in \mathcal{S}\} \subseteq \funM X$. As $(1,f_S) \in S \in \mathcal{S}$ we have $W \subseteq \bigcup \mathcal{S}$ as well. Let $S \in \mathcal{S}$, then $(1,f_S) \in S \cap W$ so that $S \cap W \neq \emptyset$. Finally, let $(o,f) \in W$. By construction, $(o,f) = (1,f_S)$ for some $S \in \mathcal{S}$; hence $o = 1$. Focus now on the converse implication $(ii) \Rightarrow (i)$. Assume there is a $W \subseteq \funM X$ that satisfies the three conditions stated in $(ii)$. Let $S \in \mathcal{S}$. By the second condition, there is some $(o,f) \in S \cap W$. By the third condition, as $(o,f) \in W$ we have $o = 1$. Hence there is indeed a $f \in X^A$ such that $(1,f) \in S$.
\end{proof}

Second step, let us prove that transitions are equal.

\begin{lemma}
The two following sets are equal:
\begin{itemize}
\item $\{V \subseteq X \mid V \subseteq \{f(a) \mid (o,f) \in S, S \in \mathcal{S}\} \text{ and } \forall S \in \mathcal{S}, \{f(a) \mid (o,f) \in S\} \cap V \neq \emptyset \}$
\item $\{\{f(a) \mid (o,f) \in W \} \mid W \subseteq \funM X \text{ s.t. } W \subseteq \bigcup \mathcal{S} \text{ and } \forall S \in \mathcal{S}, S \cap W \neq \emptyset \}$
\end{itemize}
\end{lemma}

\begin{proof}
Begin with the converse inclusion. Consider $V = \{f(a) \mid (o,f) \in W\}$ where $W \subseteq \funM X$ is such that $W \subseteq \bigcup \mathcal{S}$ and for all $S \in \mathcal{S}, S \cap W \neq \emptyset$. First we must prove that $V \subseteq \{f(a) \mid (o,f) \in S, S \in \mathcal{S}\}$. Indeed, for any $(o,f) \in W$, as $W \subseteq \bigcup \mathcal{S}$ there is some $S \in \mathcal{S}$ such that $(o,f) \in S$, whence the inclusion. Second, take $S \in \mathcal{S}$ and prove that $\{f(a) \mid (o,f) \in S\} \cap V \neq \emptyset$. This is directly obtained from the hypothesis $S \cap W \neq \emptyset$. This achieves the proof of the converse inclusion. Let us now prove that the direct inclusion holds as well. Let $V \subseteq X$ be such that $V \subseteq \{f(a) \mid (o,f) \in S, S \in \mathcal{S} \}$ and for all $S \in \mathcal{S}$, $\{f(a) \mid (o,f) \in S\} \cap V \neq \emptyset$. We have to find some $W \subseteq \funM X$ such that the three following facts hold true:
\begin{enumerate}
\item[$(i)$] $W \subseteq \bigcup \mathcal{S}$
\item[$(ii)$] for all $S \in \mathcal{S}$, $S \cap W \neq \emptyset$
\item[$(iii)$] $V = \{f(a) \mid (o,f) \in W\}$
\end{enumerate}
By using the hypothesis we can make the following constructions. For any $x \in V$, find some $S_x \in \mathcal{S}$ and some $(o_x,f_x) \in S_x$ such that $x = f_x(a)$. For any $S \in \mathcal{S}$, find some $(o_S,f_S) \in S$ such that $f_S(a) \in V$. Now define $W = \{(o_x,f_x) \mid x \in V\} \cup \{(o_S,f_S) \mid S \in \mathcal{S}\}$ and prove that it satisfies the three conditions above. For $(i)$, let $(o,f) \in W$. Either $(o,f) = (o_x,f_x)$ for some $x \in V$, and then $(o,f) \in S_x \in \mathcal{S}$ so that $(o,f) \in \bigcup \mathcal{S}$; or $(o,f) = (o_S,f_S)$ and then $(o,f) \in S \in \mathcal{S}$ from what $(o,f) \in \bigcup \mathcal{S}$ follows. This proves that $(i)$ holds true. Consider now $(ii)$. Let $S \in \mathcal{S}$. By construction, $(o_S,f_S) \in S \cap W$. Finally prove property $(iii)$ by double inclusion. Let $x \in V$. Then $(o_x,f_x) \in W$ and $f_x(a) = x$ whence $x \in \{f(a) \mid (o,f) \in W\}$. Lastly, let $(o,f) \in W$ and prove that $f(a) \in V$. Either $(o,f) = (o_x,f_x) \in S_x \in \mathcal{S}$ for some $x \in V$ and then $f(a) = f_x(a) = x \in V$; or $(o,f) = (o_S,f_S) \in S$ for some $S \in \mathcal{S}$ and then $f(a) = f_S(a) \in V$ by construction of $f_S$. This achieves the proof of $(iii)$ and also the whole proof of the lemma.
\end{proof}

These lemma jointly prove the following proposition:

\begin{proposition}
The "Yang-Baxter" diagram \ref{diag:yangbaxter} commutes.
\end{proposition}

\section{Proofs of Section~\ref{sec:gendet}}

\subsection{Proof of Proposition~\ref{prop:gendet}}

This proof consists in a slight extension of the proof of Lemma~$5.1$ of~\cite{goypetrisan}. For the sake of completeness, we explicit here the full details of the proof.

We will denote the "weak liftings" natural transformations by  $\piwl,\iotawl$ for $\wdld$ and $\piwls,\iotawls$ for $\wdls$. Let $(X,c)$ be a $\funG\funS\funT$-coalgebra. Consider the morphism
\begin{equation} \label{eq:thisisconv}
\begin{tikzcd}
X \ar[r,"c" above] & \funG\funS\funT X \ar[r, "\funG \piwl_{\funFuT X}" above] & \funG \funUuT \funSwl \funFuT X \ar[r,"\piwls_{\funSwl \funFuT X}" above] & \funUuT \funGwl \funSwl \funFuT
\end{tikzcd}
\end{equation}
and take its adjoint transpose with respect to the adjunction $\funFuT \dashv \funUuT$ to get a coalgebra $c^* : \funFuT X \to \funGwl \funSwl \funFuT X$. 
The functor $\widehat{\funFuT}$ maps $(X,c)$ to $(\funFuT X, c^*)$. Let $((X,x),d)$ be a $\funGwl \funSwl$-coalgebra in $\EMT$. The functor $\widehat{\funUuT}$ maps $((X,x),d)$ to $\funG\funS$-coalgebra
\begin{equation}
\begin{tikzcd}
X \ar[r, "\funUuT d" above] & \funUuT \funGwl \funSwl (X,x) \ar[r,"\iotawls_{\funSwl (X,x)}" above] & \funG \funUuT \funSwl (X,x) \ar[r, "\funG\iotawl_X" above] & \funG\funS X
\end{tikzcd}
\end{equation}

\subsection{Proof of Lemma~\ref{lem:same}}

The proof consists in the following diagram. Starting from the top-left node of the diagram, the path going right then down is expression  \ref{eq:samething}. The path going down then right equals the construction given in the proof of Proposition~\ref{prop:gendet}.

\adjustbox{scale=0.7,center}{%
\begin{tikzcd}[column sep=tiny]
\funT X \arrow[r, "\funT c"] \arrow[d, "\funT c"]                                                  & \funT\funG\funS\funT X \arrow[rrrr, "\wdls_{\funS\funT X}"] \arrow[d, "\funT\funG \etaT_{\funS \funT X}"]                                                         &                                                                                                       &                                                                                                                          &                                                                                                                                                                        & \funG\funT\funS\funT X \arrow[rrrr, "\funG \wdld_{\funT X}"] \arrow[rd, "\funG\funT\funS\etaTTX"]                                                                                                             &                                                                                                                 &                                                                                                                                                          &                                                              & \funG\funS\funT\funT X \arrow[dddd, "\funG\funS\muTX = \funG\funS\funUuT \varepsilon^{\monT}_{\funFuT X}"] \\
\funT\funG\funS\funT X \arrow[ddd, "\funT\funG\piwl_{\funFuT X}"'] \arrow[ru, equal] & \funT\funG\funT\funS\funT X \arrow[dd, "\funT\funG\funT\piwl_{\funFuT X}"'] \arrow[rrr, "\funUuT(\varepsilon^{\monT}\funGwl \circ \funFuT \piwls)_{\funFuT\funS\funT X}"] &                                                                                                       &                                                                                                                          & \funUuT \funGwl \funFuT \funS\funT X \arrow[rd, "\funUuT \funGwl \funFuT \funS \etaTTX"] \arrow[ru, "\iotawls_{\funFuT \funS\funT X}"] \arrow[ld, equal] &                                                                                                                                                                                                               & \funG\funT\funS\funT\funT X \arrow[rr, "\funG\funUuT(\varepsilon^{\monT}\funSwl \circ \funFuT \piwl)_{\funFuT\funT X}"] &                                                                                                                                                          & \funG\monST\funT X \arrow[ru, "\funG\iotawl_{\funS\funT X}"] &                                                                                                    \\
                                                                                                   &                                                                                                                                                                   &                                                                                                       & \funUuT \funGwl \funFuT \funS\funT X \arrow[rrdd, "\funUuT \funGwl(\varepsilon^{\monT}\funSwl \circ \funFuT\piwl)_{\funFuT X}"'] &                                                                                                                                                                        & \funUuT \funGwl \funFuT \funS \funT \funT X \arrow[ll, "\funUuT \funGwl \funFuT \funS\funUuT \varepsilon^{\monT}_{\funFuT X}"] \arrow[rr, "\funUuT \funGwl(\varepsilon^{\monT}\funSwl \circ \funFuT \piwl)_{\funFuT\funT X}"] &                                                                                                                 & \funUuT \funGwl \funSwl \funFuT \funT X \arrow[ru, "\iotawls_{\funSwl \funFuT \funT X}"] \arrow[lldd, "\funUuT \funGwl \funSwl \varepsilon^{\monT}_{\funFuT X}"] &                                                              &                                                                                                    \\
                                                                                                   & \funT\funG\funT\monST X \arrow[rd, "\funT\funG\funUuT \varepsilon^{\monT}_{\funSwl \funFuT X}"]                                                                           &                                                                                                       &                                                                                                                          &                                                                                                                                                                        &                                                                                                                                                                                                               &                                                                                                                 &                                                                                                                                                          &                                                              &                                                                                                    \\
\funT\funG\monST X \arrow[ru, "\funT\funG\etaT_{\monST X}"] \arrow[rr, equal]        &                                                                                                                                                                   & \funT\funG\monST X \arrow[rrr, "\funUuT(\varepsilon^{\monT}\funGwl \circ \funFuT \piwls)_{\funSwl\funFuT X}"] &                                                                                                                          &                                                                                                                                                                        & \funUuT \funGwl \funSwl \funFuT X \arrow[rrr, "\iotawls \funSwl \funFuT X"]                                                                                                                                   &                                                                                                                 &                                                                                                                                                          & \funG \monST X \arrow[r, "\funG \iotawl \funFuT X"]          & \funG\funS\funT X                                                                                 
\end{tikzcd}
}

On top of the diagram, the two pentagons commute because of the expressions of $\wdls$ and $\wdld$ obtained via the proof of Proposition~\ref{prop:weakbij}. The top-left triangle commutes trivially. All the other polygons commute because of naturality, except two triangles that use the adjunction property $\funUuT \varepsilon^{\monT} \circ \etaT \funUuT = 1$. 

\subsection{Proof of Proposition~\ref{prop:goodsemantics}}

The proof is elementary and as we said, already sketched in~\cite{klin2015coalgebraic}. Let us recall all necessary ingredients, in particular the usual semantics of non-deterministic automata.

\begin{align}
& \semAA{x}(\varepsilon) = o(x) & \semAA{x}(aw) = \bigvee_{U \in t_a(x)} \bigwedge_{y \in U} \semAA{y} (w) \\
& \semNDA{U}(\varepsilon) = o^+(U) & \semNDA{U}(aw) = \bigvee_{W \in t_a^+(x)} \semNDA{W}(w)
\end{align}
where
\begin{align}
& o^+(U) = \bigwedge_{x \in U} o(x) \\
& t_a^+(U) = \left\{\bigcup_{x \in U} K_x \mid \forall x \in U, K_x \in \unions(t_a(x))\right\}
\end{align}

The proof is by induction on $w$. The basic case $w = \varepsilon$ is trivial given the above definitions. Assume
\begin{equation}
\forall U \in \funP X, \semNDA{U}(w) = \bigwedge_{y \in U} \semAA{y}(w)
\end{equation}

Then fix a $U \in \funP X$ and compute the two quantities for $aw$:

\begin{align}
& \semNDA{U}(aw) = \bigvee_{V \in t_a^+(U)} \semNDA{V}(w) = \bigvee_{V \in t_a^+(U)} \bigwedge_{x \in V} \semAA{x}(w) \\
& \bigwedge_{y \in U} \semAA{y}(aw) = \bigwedge_{y \in U} \bigvee_{W \in t_a(y)} \bigwedge_{x \in W} \semAA{x}(w)
\end{align}

These quantities are equal.

Assume the first quantity equals $1$. There is $V \in t_a^+(U)$ such that for all $x \in V$, $\semAA{x}(w) = 1$.
One can find for every $y \in U$ a $K_y \in \unions{t_a(y)}$ such that $V = \bigcup_{y \in U} K_y$. Fix some $y \in U$. By the definition of $\unions$, there
is a $W \in t_a(y)$ such that $W \subseteq K_y$. Let $x \in W$, then $x \in K_y$ so $x \in V$, which yields $\semAA{x}(w) = 1$ by hypothesis.

Assume the second quantity equals $1$. For every $y \in U$ there is $W \in t_a(y)$ such that for all $x \in W$, $\semAA{x}(w) = 1$. Take $K_y$ to be such a $W$, then $K_y \in t_a(y) \subseteq \unions(t_a(y))$ so that we can define $V = \bigcup_{y \in U} K_y \in t_a^+(U)$. Let $x \in V$. By construction of $V$, $x \in K_y$ for some $y \in U$. Hence $\semAA{x}(w) = 1$ by hypothesis on $K_y$. This achieves the proof.

\section{Proofs of Section~\ref{sec:upto}}

\subsection{Proof of Proposition~\ref{lem:ourbialgebra}}

The proof that $(\funT X, \muTX, c^+)$ is a $\funG\wdld \circ \wdls\funS$-bialgebra consists in the following diagram. Actually, this is the same diagram occuring with standard "distributive laws", because the expression of $c^+$ is the same in both frameworks and units of monads are not involved.

\begin{tikzcd}[column sep=tiny]
\funT\funT X \arrow[rd, "\funT\funT c"] \arrow[r, "\muTX"] \arrow[dddd, "\funT c^+"'] & \funT X \arrow[rd, "\funT c"] \arrow[rrrrr, "c^+"]                                                          &                                                          &                                                                                                &                                                                                                          &                                                                                               & \funG\funS\funT X                                      \\
                                                                                      & \funT\funT\funG\funS\funT X \arrow[d, "\funT\wdls_{\funS\funT X}"'] \arrow[r, "\muTGSTX"']                  & \funT\funG\funS\funT X \arrow[r, "\wdls_{\funS\funT X}"] & \funG\funT\funS\funT X \arrow[rr, "\funG  \wdld_{\funT X}"]                                    &                                                                                                          & \funG\funS\funT\funT X \arrow[ru, "\funG\funS\muTX"]                                          &                                                        \\
                                                                                      & \funT\funG\funT\funS\funT X \arrow[d, "\funT\funG \wdld_{\funT X}"] \arrow[rr, "\wdls_{\funT\funS\funT X}"] &                                                          & \funG\funT\funT\funS\funT X \arrow[u, "\funG\muTSTX"] \arrow[r, "\funG\funT \wdld_{\funT X}"'] & \funG\funT\funS\funT\funT X \arrow[r, "\funG \wdld_{\funT\funT X}"'] \arrow[ldd, "\funG\funT\funS\muTX"] & \funG\funS\funT\funT\funT X \arrow[u, "\funG\funS\muTTX"] \arrow[rdd, "\funG\funS\funT\muTX"] &                                                        \\
                                                                                      & \funT\funG\funS\funT\funT X \arrow[ld, "\funT\funG\funS \muTX"]                                             &                                                          &                                                                                                &                                                                                                          &                                                                                               &                                                        \\
\funT\funG\funS\funT X \arrow[rrr, "\wdls_{\funS\funT X}"]                            &                                                                                                             &                                                          & \funG\funT\funS\funT X \arrow[rrr, "\funG \wdld_{\funT X}"]                                    &                                                                                                          &                                                                                               & \funG\funS\funT\funT X \arrow[uuuu, "\funG\funS\muTX"]
\end{tikzcd}

The center pentagons commute because of the $(\muT)$ diagrams of the "weak distributive laws" $\wdld$ and $\wdls$. Going left and then counter-clockwise, polygons commute by definition of $c^+$, naturality of $\wdls$, naturality of $\wdld$, monad property of $\muT$, definition of $c^+$, naturality of $\muT$.

\subsection{Proof of the diagram in Remark~\ref{rem:wdlmorphism}}
First note that $\supp$ is natural because for any function $f : X \to Y$, reminding that $\funD f(\phi)(y) = \sum_{x \in f^{-1}(\{y\})} \phi(x)$ we have
\begin{align}
(\supp_Y \circ \funD f)(\phi) & = \{y \mid \supp_X(\phi) \cap f^{-1}(\{y\}) \neq \emptyset \} \\
& = \{ f(x) \mid x \in \supp_X(\phi)\} = (\funP f \circ \supp_X)(\phi)
\end{align}
Let $\wdld : \monD \monP \Rightarrow \monP \monD$ and $\wdld' : \monP \monP \Rightarrow \monP \monP$ as in Examples~\ref{ex:pd} and~\ref{ex:pp}. Take $\Phi = \sum_{i\in I} p_i A_i \in \funD\funP X$, where the $A_i$ are distinct, $p_i > 0$ and $\sum_{i\in I} p_i = 1$ (hence $I \neq \emptyset)$. Let us compute both paths:
\begin{align}
(\wdld' \circ \supp \funP)_X(\Phi) &= \left\{\wdld'_X(\{A_i\mid i \in I\})\right\} \\
&= \left\{B \subseteq X \mid B \subseteq \bigcup_{i\in I} A_i \text{ and } \forall i \in I, B \cap A_i \neq \emptyset \right\} \\
(\funP\supp \circ \wdld)_X(\Phi) &= \funP\supp_X \left(\left\{\sum_{i\in I} p_i \phi_i \mid \forall i \in I, \supp(\phi_i) \subseteq A_i\right\}\right)\\
& = \left\{ \bigcup_{i\in I} \supp(\phi_i) \mid \forall i \in I, \supp(\phi_i) \subseteq A_i \right\}
\end{align}
Note that in the above expression the $\phi_i$ are not necessarily distinct. Let us prove that these two sets are the same. Let $B \subseteq X$ such that $B \subseteq \bigcup_{i\in I} A_i$ and for all $i \in I$, $B \cap A_i \neq \emptyset$. Denote $B \cap A_i = \{x_1^i,...,x_{n_i}^i\}$ with $n_i \geq 1$. Define for all $i \in I$ the distribution $\phi_i = \sum_{k = 1}^{n_i} \frac{1}{n_i} x_k^i$, then $\sum_{i\in I} \frac{1}{|I|} \phi_i$ is a distribution because $I \neq \emptyset$. We have $\supp_X(\phi_i) = B \cap A_i \subseteq A_i$ and $\bigcup_{i\in I} \supp_X(\phi_i) = B \cap \bigcup_{i\in I} A_i = B$. For the converse inclusion, let $\phi_i \in \funD X$ such that $\supp_X(\phi_i) \subseteq A_i$. Let us prove that $\bigcup_{i\in I} \supp_X(\phi_i)$ satisfies the conditions of the first set. Indeed, it is clear that $\bigcup_{i\in I} \supp_X(\phi_i) \subseteq \bigcup_{i\in I} A_i$. Let $i_0 \in I$, $\left(\bigcup_{i\in I} \supp_X(\phi_i)\right) \cap \supp_X(\phi_{i_0}) = \supp_X(\phi_{i_0}) \neq \emptyset$. This achieves the proof.

\end{document}